%% file: random_facet_shortest_paths.tex
\documentclass[twoside]{article}

\usepackage{url}
\usepackage{amsmath}
\usepackage{amssymb}
\usepackage[boxruled,vlined,nokwfunc]{algorithm2e}

\usepackage{graphics,graphicx}
\usepackage{color}
\usepackage{multirow}

\newcommand{\TODO}{\textcolor{red}{TODO}}

\usepackage{ifpdf}
\ifpdf
\else
\fi

\include{./macros_random_facet_shortest_paths}

\begin{document}

\title{\RandomFacet\ and \RandomBland\\ require subexponential time even for shortest paths
}

\author{Oliver Friedmann\thanks{Department of Computer Science,
University of Munich, Germany. E-mail: {\tt
  Oliver.Friedmann@gmail.com}.}\\
 \and
Thomas Dueholm Hansen\thanks{Department of Management Science and Engineering, Stanford University, USA. Supported by The Danish Council for Independent Research $|$ Natural Sciences (grant no. 12-126512). E-mail:
{\tt tdh@cs.au.dk}.}\\
\and
Uri Zwick\thanks{Blavatnik School of Computer Science, Tel Aviv University,
  Israel. Research supported by BSF grant no. 2012338 and by the
The Israeli Centers of Research Excellence (I-CORE) program,
(Center No. 4/11). E-mail: {\tt
	zwick@tau.ac.il}.} }
\date{}

\maketitle

\begin{abstract}
The \RandomFacet\ algorithm of Kalai and of Matou{\v{s}}ek, Sharir and Welzl is an elegant randomized algorithm for solving linear programs and more general LP-type problems.
Its expected subexponential time of $2^{\tilde{O}(\sqrt{m})}$, where $m$ is the number of inequalities, makes it the fastest known combinatorial algorithm for solving linear programs. We previously showed that \RandomFacet\ performs an expected number of $2^{\tilde{\Omega}(\sqrt[3]{m})}$ pivoting steps on some LPs with~$m$ inequalities that correspond to $m$-action \emph{Markov Decision Processes} (MDPs). We also showed that \RandomPerm, a one permutation variant of \RandomFacet, performs an expected number of $2^{\tilde{O}(\sqrt{m})}$ pivoting steps on these examples. Here we show that the same results can be obtained using LPs that correspond to instances of the classical \emph{shortest paths} problem. This shows that the stochasticity of the MDPs, which is essential for obtaining lower bounds for \RandomEdge, is not needed in order to obtain lower bounds for \RandomFacet. We also show that our new $2^{\tilde{\Omega}(\sqrt{m})}$ lower bound applies to \RandomBland, a randomized variant of the classical anti-cycling rule suggested by Bland.
\end{abstract}


\input{introduction}
\input{RandomFacet}

\input{Simplified-Construction}

\input{Actual-Construction}

\input{Lower-RandomFacet}

\input{Lower-RandomFacet-star}
\input{Lower-RandomBland}

\input{Conclude}

\section*{Acknowledgement}

We would like to thank Bernd G{\"a}rtner, G\"{u}nter Rote and Tibor Szab\'{o} for various discussions on the \RandomFacet\ algorithm and its variants that helped us realize that the expected number of pivoting steps performed by \RandomFacet\ and \RandomPerm\ are \emph{not} the same.

\bibliographystyle{abbrv}
\bibliography{./random_facet_shortest_paths}

\appendix

\newpage
\section{Proof of Lemma~\ref{L-recurrence}}\label{A-recurrence}

\begin{proof}
We must show that:
\[
\displaystyle f(n) ~=~ \sum_{k=1}^n \frac{1}{k!} {n \choose k} \;.
\]
Observe that:
\[
\frac{1}{n} \sum_{i=0}^{n-1} \sum_{k=1}^{i} \frac{1}{k!} {i \choose k} ~=~
\frac{1}{n} \sum_{k=1}^{n-1} \frac{1}{k!} \sum_{i=k}^{n-1} {i \choose k} ~=~
\frac{1}{n} \sum_{k=1}^{n-1} \frac{1}{k!} {n \choose k+1} ~=~
\sum_{k=1}^{n-1} \frac{1}{(k+1)!} {n-1 \choose k} ~=~
\sum_{k=2}^{n} \frac{1}{k!} {n-1 \choose k-1} \;.
\]
Then by induction:
\begin{align*}
f(n) \;&=\; f(n-1) + 1 + \frac{1}{n} \sum_{i=0}^{n-1} f(i)\\
\;&=\; \sum_{k=1}^{n-1} \frac{1}{k!} {n-1 \choose k} + 1 + \frac{1}{n} \sum_{i=0}^{n-1} \sum_{k=1}^{i} \frac{1}{k!} {i \choose k} \\
\;&=\; \sum_{k=1}^{n-1} \frac{1}{k!} {n-1 \choose k} + 1 + \sum_{k=2}^{n} \frac{1}{k!} {n-1 \choose k-1} \\
\;&=\; \sum_{k=1}^{n} \frac{1}{k!} {n-1 \choose k} + \sum_{k=1}^{n} \frac{1}{k!} {n-1 \choose k-1} \\
\;&=\; \sum_{k=1}^{n} \frac{1}{k!} {n \choose k} ~.
\end{align*}
\end{proof}

\section{Proof of Lemma~\ref{lemma:counters}}\label{A-counters}

\begin{proof}
We must show that $f(n) = f^{1P}(n)$.

The lemma is proved by using induction and linearity of
expectation. For $n=0$ we have $f(0) = f^{1P}(0) = 0$.
Let $\Sigma(N)$, for $N \subseteq [n]$, be the set of permutations of
$N$, i.e., every $\sigma \in \Sigma(N)$ is a map $\sigma: N \to [|N|]$.
Note that $|\Sigma(N)| = |N|!$. Note also that $f^{1P}(N,\sigma) =
f^{1P}(N,\sigma')$ where $N\subseteq [n]$ and $\sigma \in \Sigma([n])$,
and where $\sigma' \in \Sigma(N)$ is obtained by compressing $\sigma$.
We see that:
\begin{align*}
f^{1P}(n) ~&=~ \frac{1}{n!} \sum_{\sigma \in \Sigma([n])} f^{1P}([n],\sigma)\\
~&=~ \frac{1}{n} \sum_{i \in [n]} \frac{1}{(n-1)!} \sum_{\sigma \in
	\Sigma([n]\setminus\{i\})} f^{1P}([n]\setminus\{i\},\sigma) + 1 + f^{1P}([i-1],\sigma) \\
~&=~ \frac{1}{n} \sum_{i \in [n]} f^{1P}(n-1) + 1 + f^{1P}(i-1) \\
~&=~ f(n-1) + 1 + \frac{1}{n} \sum_{i \in [n]} f(i-1) \\
~&=~ f(n)\;.
\end{align*}
\end{proof}

\end{document}

%% file: macros_random_facet_shortest_paths.tex
\newtheorem{theorem}{Theorem}[section]
\newtheorem{lemma}[theorem]{Lemma}
\newtheorem{claim}[theorem]{Claim}

\newtheorem{definition}[theorem]{Definition}

\newenvironment{proof}{{\bf Proof:\ }}{\hfill$\Box$\medskip}
\newcommand{\qed}{\hfill$\Box$\medskip}

\newcommand{\RR}{{\mathbb R}}

\def\NN{\mathbb N}

\DeclareMathOperator*{\argmin}{argmin}

\newcommand{\etal}{{\em et al.\ }}
\newcommand{\mathsc}[1]{\mbox{\textsc #1}}
\newcommand{\id}{1\hspace{-1,5ex}1}

\newcommand{\indexfunction}{\mathtt{ind}}

\newcommand{\ind}{\indexfunction}
\newcommand{\bit}{\mathtt{bit}}

\newcommand{\aaa}{{\bf a}}
\newcommand{\bbb}{{\bf b}}
\newcommand{\uuu}{{\bf u}}
\newcommand{\www}{{\bf w}}
\newcommand{\eee}{{\bf e}}

\newcommand{\BB}{{\mathcal{B}}}

\newcommand{\TT}{{\mathsc t}}

\newcommand{\RandomFacet}{\mbox{\sc Random-}\allowbreak\mbox{\sc Facet}}
\newcommand{\RandomBland}{\mbox{\sc Random-}\allowbreak\mbox{\sc Bland}}
\newcommand{\Bland}{\mbox{\sc Bland}}
\newcommand{\RandomEdge}{\mbox{\sc Random-}\allowbreak\mbox{\sc Edge}}
\newcommand{\RandomPerm}{{\mbox{\sc Random-}\allowbreak\mbox{$\mbox{\sc Facet}^{1P}$}}}

\newcommand{\RandCountP}{\mbox{$\mbox{\sc RandCount}^{1P}$}}

\newcommand{\RandCount}{\mbox{\sc RandCount}}
\newcommand{\NonRecursiveRandCount}{\mbox{\sc NonRecursiveRandCount}}

\newcommand{\val}{\mbox{y}}

\newcommand{\E}[1]{\mathbb{E}\left[{#1}\right]}
\newcommand{\PPr}[1]{\mathbb{P}\left[{#1}\right]}

\newcommand{\ignore}[1]{}

%% file: introduction.tex
\section{Introduction}

Linear programming (LP) is one of the most successful mathematical modeling tools. The \emph{simplex algorithm}, introduced by Dantzig \cite{Dantzig63}, is one of the most widely used methods for solving linear programs. The simplex algorithm starts at a vertex of the polytope corresponding to the linear program. (We assume, for simplicity, that the linear program is feasible, bounded, and non-degenerate, and that a vertex of the polytope is available.) If the current vertex is not optimal, then at least one of the edges incident on it leads to a neighboring vertex with a smaller objective function. A \emph{pivoting rule} determines which one of these vertices to move to. The simplex algorithm, with any pivoting rule, is guaranteed to find an optimal solution of the linear program.

Unfortunately, with essentially all known deterministic pivoting rules, the simplex method is known to require \emph{exponential time} on some linear programs (see, e.g., Klee and Minty \cite{KlMi72}, Amenta and Ziegler \cite{AmZi96}, and a recent subexponential bound of Friedmann \cite{Friedmann/IPCO11}). While there are other polynomial time algorithms for solving LP problems, most notably the \emph{ellipsoid algorithm} (Khachian \cite{Khachiyan79}) and \emph{interior point methods} (Karmarkar \cite{Karmarkar84}), these algorithms are not \emph{strongly} polynomial, i.e., their running time, in the unit-cost model, depends on the number of bits needed to represent the coefficients of the LP, and not just on the combinatorial size of the problem, i.e., the number of variables and the number of constraints. The question of whether there exists a strongly polynomial time algorithm for solving linear programs is of great theoretical importance.

Kalai \cite{Kalai92,Kalai97} and Matou{\v{s}}ek, Sharir and Welzl \cite{MaShWe96} devised a \emph{randomized} pivoting rule for the simplex algorithm, known as \RandomFacet, and obtained a \emph{subexponential} $2^{O(\sqrt{(m-d)\log d})}$ upper bound on the expected number of pivoting steps it performs on \emph{any} linear program, where $d$ is the \emph{dimension} of the linear program, and~$m$ is the number of inequality constraints. Matou{\v{s}}ek, Sharir and Welzl \cite{MaShWe96} actually obtained a dual version of the \RandomFacet\ pivoting rule and obtained an upper bound of $2^{O(\sqrt{d\log(m-d)})}$ on the number of dual pivoting steps it performs.
\RandomFacet\ is currently the fastest known pivoting rule for the simplex algorithm.

The \RandomFacet\ pivoting rule works as follows. Let~$v$ be the current vertex visited by the simplex algorithm. Randomly choose one of the \emph{facets} of the polytope that contain~$v$. Let $F$ be the facet chosen. Recursively find the optimal vertex~$v'$ among all vertices of the polytope that lie on~$F$. This corresponds to finding an optimal solution for a modified linear program in which the inequality corresponding to~$F$ is replaced by an equality. If $v'$ is also an optimal solution of the original problem, we are done. Otherwise, it is not difficult to check that there must be a single edge leading from $v'$ to a vertex $v''$ with a smaller value. This edge is taken and the algorithm is recursively run from~$v''$. 
A more formal description, as well as an equivalent non-recursive description, and descriptions of two variants, $\RandomPerm$ and $\RandomBland$, of \RandomFacet\ are given in Section~\ref{S-Random-Facet}.

\RandomFacet\ can be used to solve not only linear programs, but also a wider class of abstract optimization problems known as \emph{LP-type} problems (see \cite{MaShWe96}). Matou{\v{s}}ek \cite{Matousek94} showed that the subexponential upper bound on the complexity of \RandomFacet\ is essentially tight for LP-type problems. The instances used by Matou{\v{s}}ek~\cite{Matousek94}, however, are far from being linear programs.

In~\cite{FriedmannHansenZwick/SODA11,FriedmannHansenZwick/STOC11}, we claimed an $2^{\tilde{\Omega}(\sqrt{m})}$ lower bound on the complexity of \RandomFacet\ for actual LPs. The lower bound presented there, however, is actually for the \emph{one-permutation} variant $\RandomPerm$ of \RandomFacet. In~\cite{FriedmannHansenZwick/SODA11}, we erroneously claimed that the expected running times of $\RandomFacet$ and $\RandomPerm$ are the same, and thus thought that our lower bound also applied to $\RandomFacet$. 
The mistake is pointed out in \cite{Hansen12,FrHaZw14}. The lower bounds given in~\cite{FriedmannHansenZwick/SODA11,FriedmannHansenZwick/STOC11} for \RandomPerm\ are adapted in \cite{Hansen12} to give an $2^{\tilde{\Omega}(\sqrt[3]{m})}$ lower bound for $\RandomFacet$.

In this paper, we obtain an $2^{\tilde{\Omega}(\sqrt[3]{m})}$ lower bound for the original $\RandomFacet$ pivoting rule. While the new bound is smaller then the bounds obtained in~\cite{FriedmannHansenZwick/SODA11,FriedmannHansenZwick/STOC11} for $\RandomPerm$, a square-root in the exponent is replaced by a cube-root, it still provides a subexponential lower bound on the running time of \RandomFacet. We also obtain a new $2^{\tilde{\Omega}(\sqrt{m})}$ lower bound on the running time of \RandomBland, a randomized version of Bland's pivoting rule (Bland \cite{Bland77}), which is closely related to both $\RandomFacet$ and $\RandomPerm$.


The lower bound given in~\cite{FriedmannHansenZwick/SODA11} is for \emph{parity games}, a class of deterministic 2-player games. The lower bounds given in~\cite{FriedmannHansenZwick/STOC11,Hansen12} are for LPs that correspond to \emph{Markov Decision Processes} (MDPs), a class of stochastic 1-player games. (For more on MDPs, see Puterman \cite{Puterman94}.)
Here, we obtain subexponential lower bounds for the running times of \RandomFacet, \RandomPerm\ and \RandomBland\ on LPs that correspond to standard acyclic \emph{shortest paths} problems, the simplest form of deterministic 1-player games.

It is interesting, and perhaps surprising, that \RandomFacet, the fastest known pivoting rule for general linear programs, requires a subexponential number of pivoting steps even for acyclic shortest paths problems which can be easily solved in linear time.

In \cite{FriedmannHansenZwick/STOC11} we also presented MDPs on which \RandomEdge, a different randomized pivoting rule, requires an expected number of $2^{\tilde{\Omega}(\sqrt[4]{m})}$ pivoting steps. In contrast with \RandomFacet, it is not difficult to show that \RandomEdge\ solves shortest paths problems in expected polynomial time.

As mentioned, \RandomFacet\ can be used to solve not only linear programs, but also a wider class of problems known as LP-type problems (see Matou{\v{s}}ek \etal\ \cite{MaShWe96}). Ludwig \cite{Ludwig95} was the first to show that \RandomFacet\ can be used to solve \emph{stochastic games}. Petersson and Vorobyov \cite{PeVo01}, and Bj{\"o}rklund \etal\ \cite{BjVo05,BjVo07} used \RandomFacet\ to solve several deterministic and stochastic games. Halman \cite{Halman07} showed that all these form LP-type problems.

The line of work pursued here started by Friedmann \cite{Friedmann09,Friedmann11} who proved that the \emph{strategy iteration} algorithm for \emph{parity games} may require exponential time. The strategy iteration algorithm is an adaptation to 2-player games of Howard's \cite{Howard60} \emph{policy iteration} algorithm for solving MDPs. (For more on parity games, see also
V{\"{o}}ge and Jurdzi{\'{n}}ski \cite{VoJu00} and Jurdzi{\'n}ski \etal \cite{JuPaZw08}.). Fearnley \cite{Fearnley10} used a similar construction to obtain an exponential lower bound for Howard's algorithm for MDPs, a class of stochastic 1-player games. Policy iteration algorithms are similar in nature to the simplex algorithm. The main difference is that they may apply \emph{multi-switches}, i.e., perform several pivoting steps simultaneously. In~\cite{FriedmannHansenZwick/SODA11,FriedmannHansenZwick/STOC11} and here we continue this line of work and obtain lower bounds for \RandomFacet\ applied to shortest paths problems.
Although \RandomFacet\ performs only one pivoting step at a time, its randomized nature makes it similar to more general policy iteration algorithms.

The rest of the paper is organized as follows. In Section~\ref{S-Random-Facet} we give a more formal description of the \RandomFacet\ pivoting rule and its variants. In Section~\ref{S-short-path} we remind the reader how shortest paths problems can be cast as linear programs and solved by the simplex algorithm. We also explain how \RandomFacet\ works in this concrete setting. In Section~\ref{S-randomized-counter} we describe a \emph{randomized counter} on which our lower bounds, as well as the lower bounds in~\cite{FriedmannHansenZwick/SODA11,FriedmannHansenZwick/STOC11} are based. In Section~\ref{sec:construction} we describe instances of the shortest paths problems that can be used to simulate the behavior of the randomized counter. In Section~\ref{S-lower-bound} we obtain the lower bound for \RandomFacet. In Section~\ref{S-one-pass-variant} we obtain the lower bound of \RandomPerm. In Section~\ref{S-randomized-blands-rule} we obtain the lower bound of \RandomBland. We end in Section~\ref{S-concluding-remarks} with some concluding remarks and open problems.



%% file: RandomFacet.tex
\section{\RandomFacet, $\RandomPerm$ and \RandomBland}\label{S-Random-Facet}

\newcommand{\policyIteration}{\mbox{\sc PolicyIteration}}
\newcommand{\AAA}{{\bf A}}
\newcommand{\BBB}{{\bf B}}
\newcommand{\cc}{{\bf c}}
\newcommand{\bcc}{\bar{\bf c}}
\newcommand{\bb}{{\bf b}}
\newcommand{\xx}{{\bf x}}
\newcommand{\JJ}{{\bf J}}
\newcommand{\ee}{{\bf e}}
\newcommand{\0}{{\bf 0}}
\newcommand{\xxx}{{\bf \bar{x}}}
\newcommand{\ttt}{{\textsc t}}
\newcommand{\yy}{{\bf y}}

\newcommand{\REDUCED}{\mbox{\sc Reduced-Cost}}
\newcommand{\PIVOT}{\mbox{\sc Pivot}}
\newcommand{\RANDOM}{\mbox{\sc Random}}
\newcommand{\IMPROVE}{\mbox{\sc Improve}}

We begin with a brief introduction to linear programming problems and the simplex algorithm. For a more thorough treatment, the reader is refereed to the textbooks of Dantzig \cite{Dantzig63}, Chv{\'a}tal \cite{Chvatal83}, Schrijver \cite{Schrijver86}, Bertsimas and Tsitsiklis \cite{Bertsimas1997introduction}, and Matou{\v{s}}ek and G{\"a}rtner \cite{MaGa07}.
As customary, we consider linear programming problems in standard form:
\[
(P)~
\begin{array}{llll}
\min & \cc^T \xx&& \\
\mbox{s.t.} &\AAA \xx &=& \bb \\
     &\xx&\ge& \0
\end{array} \quad\quad\quad\quad
(D)~
\begin{array}{llll}
\max & \bb^T \yy&& \\
\mbox{s.t.} &\AAA^T \yy &\le& \cc \\
\end{array}
\]
where $\AAA\in \RR^{m\times n}$, $\bb\in \RR^m$ and $\cc\in \RR^n$. Here~$n$ is the number of variables of the linear program, and~$m$ is the number of linear equalities. We assume that the rows of~$A$ are linearly independent. The \emph{dimension} of the linear program is defined to be $d=n-m$.
The linear program $(P)$ on the left is referred to as the \emph{primal} linear program, and the linear program $(D)$ on the right is referred to as the \emph{dual} linear problem. (We focus, here, mainly on the primal linear program.)

Let $B\subset[n]=\{1,2,\ldots,n\}$, $|B|=m$. We let $\BBB=\AAA_B\in \RR^{m\times m}$ be the $m\times m$ matrix obtained by selecting the columns of~$A$ whose indices belong to~$B$. If the columns of $\BBB$ are linearly independent, we refer to~$B$ as a \emph{basis} and to~$\BBB$ as a \emph{basis matrix}. There is then a unique vector $\xx\in\RR^n$ for which $\AAA\xx=\bb$ and $\xx_i=0$, for $i\not\in B$. If we let $N=[m]\setminus B$, then $\xx_B=\BBB^{-1}\bb$ and $\xx_N=\0$. This vector~$\xx$ is the \emph{basic solution} corresponding to~$B$. If $\xx_B\ge \0$, then~$\xx$ is said to be a \emph{basic feasible solution} (bfs). A bfs is a \emph{vertex} of the polyhedron corresponding to~$(P)$. The variables in $\{\xx_i \mid i\in B\}$, are referred to a \emph{basic} variables, while the variables in $\{\xx_i\mid i\in N\}$, are referred to as \emph{non-basic} variables. A simple and standard manipulation shows that the objective function can also be expressed as $\cc^T\xx = \cc_B^T\xx_B + \bcc^T\xx$, where $\bcc=\cc-(\cc_B^T\BBB^{-1}\AAA)^T$. The vector $\bcc\in\RR^n$ is referred to as the vector of \emph{reduced costs}. Note that $\cc_B^T\xx_B$ here is a constant. Also note that $\bcc_B=\0$, i.e., the reduced costs of the basic variables are all $0$. It is also not difficult to check that if $\bcc\ge \0$, then $\xx$ is an \emph{optimal} solution.

The simplex algorithm starts with a bfs $\xx$ corresponding to a basis~$B$. (We do not get here into the question as to how the first bfs is obtained.) If $\bcc\ge \0$, then $\xx$ is an optimal solution, and we are done. Otherwise, there must be an index $i\in N=[m]\setminus B$ for which $\bcc_i<0$. If the linear program is \emph{non-degenerate}, then the set $\{\xx\in \RR^n \mid \AAA\xx=\bb , \xx_j=0 \mbox{ for $j\in N\setminus\{i\}$} \}$ is an \emph{edge} of the polytope corresponding to~$(P)$ that leads to a bfs~$\xx'$ corresponding to a basis $B'=B\cup\{i\}\setminus\{j\}$, for some $j\in B$.  Furthermore, $\cc^T\xx'<\cc^T\xx$. The algorithm moves from~$\xx$ to~$\xx'$ and continues from there. This is referred to as a \emph{pivoting step}. If there are several indices~$i$ for which $\bcc_i<0$, a \emph{pivoting rule} is used to select one of them.

Dantzig's \cite{Dantzig63} pivoting rule, chooses an index~$i$ with the \emph{smallest reduced cost}, i.e., an index~$i$ which minimizes~$\bcc_i$. Bland's \cite{Bland77} pivoting rule chooses the \emph{smallest} index~$i$ for which $\bcc_i<0$. Bland's pivoting rule ensures the termination of the simplex algorithm even in the presence of degeneracies.

\RandomFacet\ is a more complicated randomized pivoting rule introduced by Kalai \cite{Kalai92,Kalai97}, and in a dual form by Matou{\v{s}}ek, Sharir and Welzl \cite{MaShWe96}. Unlike Dantzig's and Bland's rules mentioned above, \RandomFacet\ has \emph{memory}. Its behavior depends on decisions made at  previously visited bfs's, and not only on the current bfs. We begin by describing a recursive version of \RandomFacet\ which closely follows~\cite{Kalai92,Kalai97}. Let~$\xx$ be the initial bfs, corresponding to basis~$B$. Select a random index $i\not\in B$. Recursively find an optimal solution $\xx'$ of the linear program in which the inequality constraint $\xx_i\ge 0$ is replaced by the equality constraint $\xx_i=0$. (This corresponds to staying within the \emph{facet} $\xx_i=0$.) This effectively \emph{removes} the non-basic variable $\xx_i$, the column corresponding to it in~$\AAA$, and the entry corresponding to it in~$\cc$ from the linear program. If $\xx'$ is also an optimal solution of the original linear program, we are done. Otherwise, $i$ is the only index for which $\bcc'_i<0$, where $\bcc'$ is the vector of reduced costs corresponding to~$\xx'$. Perform the pivoting step in which $i$ enters the set of basic indices, and obtain a new bfs $\xx''$. Run the algorithm recursively from~$\xx''$.

Pseudo-code of the recursive version of \RandomFacet\ is given on the top of Figure~\ref{F-RandomFacet}. The first argument~$F$ of \RandomFacet\ is the set of indices of the variables that are \emph{not} constrained to be~$0$. Initially $F=[n]$. The second argument~$B$ is the current basis. Whenever \RandomFacet\ is called, it is assumed that~$B$ defines a valid bfs of the problem and that $B\subseteq F$. For simplicity, the pseudo-code given works exclusively with bases. It returns a basis~$B$ that corresponds to an optimal bfs of the problem. If $F=B$, then $B$ is
the only, and therefore optimal, solution of the current problem, so it is returned. Otherwise, \RandomFacet\ chooses at random an index~$i$ of a currently non-basic variable. By the recursive call $B'\gets \RandomFacet(F\setminus\{i\},B)$ it computes an optimal solution under the additional constraint that the $i$-th variable is required to be~$0$. The call $\IMPROVE(B',i)$ checks whether performing a pivot step that inserts~$i$ into the basis~$B'$ would yield a bfs with a smaller objective function. More specifically, $\IMPROVE(B',i)$ computes the bfs~$\xx'$ corresponding to~$B'$ and its reduced cost vector~$\bcc'$, and checks whether $\bcc'_i<0$. If so, $B''\gets \PIVOT(B',i)$ performs the corresponding pivot step and returns the resulting basis~$B''$. (Note that $B''=B\cup\{i\}\setminus\{j\}$, for some $j\in B$.) Finally, the second recursive call $\RandomFacet(F,B'')$ computes and returns an optimal basis of the problem.

The description of \RandomFacet\ given on the top of Figure~\ref{F-RandomFacet} is identical to the description of \RandomFacet\ given by Matou{\v{s}}ek, Sharir and Welzl \cite{MaShWe96}, with the only difference that Matou{\v{s}}ek \etal\ consider it to be an algorithm for solving the dual linear program $(D)$. Note that every basis $B\subseteq [n]$, $|B|=m$ constructed by the algorithm defines not only a basic feasible primal solution, as described above, but also a basic, not necessarily feasible, dual solution $\yy=(\BBB^T)^{-1}\cc_B$. In fact, all dual solutions, except the last one, are infeasible. Note that $\bcc=\cc-\AAA^T \yy$. Choosing an index $i\in F\setminus B$ and adding the primal constraint $\xx_i=0$, corresponds to \emph{discarding} the $i$-th dual constraint. If the solution returned by the first recursive call $\RandomFacet(F\setminus\{i\},B)$ is dual feasible, then it is an optimal solution, of both the primal and dual linear programs. Otherwise, a pivoting step is performed. The correspondence between the algorithms of Kalai \cite{Kalai92,Kalai97} and of Matou{\v{s}}ek, Sharir and Welzl~\cite{MaShWe96} was first noticed by Goldwasser \cite{Goldwasser95}.

\RandomFacet\ makes a fresh random choice at each invocation. It is natural to wonder whether so much randomness is really needed. Two \emph{one-permutation} variants of $\RandomFacet$ are given at the bottom of Figure~\ref{F-RandomFacet}. Both of them base their choices on a single permutation~$\sigma$ of~$[n]$. This permutation is randomly chosen before the first invocation. The same (random) permutation is then used in all recursive calls. The first of these variants, shown on the bottom left, referred to as \RandomPerm, is identical to \RandomFacet\ in every aspect, except that instead of choosing a random $i\in F\setminus B$, it chooses the index~$i$ from $F\setminus B$ for which $\sigma(i)$ is minimized. The second variant, shown on the bottom right of Figure~\ref{F-RandomFacet}, contains an additional variation. Instead of choosing the index~$i$ from $F\setminus B$, it is chosen from the whole of~$F$. Also, the recursion bottoms out when $F=\emptyset$, and not when $F=B$. The invariant $B\subseteq F$ is no longer maintained.

We refer to the second variant as \Bland\ as, as we shall see, it is actually equivalent to Bland's \cite{Bland77} pivoting rule. When $\sigma$ is a uniformly random permutation we refer to the algorithm as $\RandomBland$.

It is not difficult to verify the correctness of these two variants. $\RandomPerm(F,B,\sigma)$ finds an optimal solution when all variables not in~$F$ are required to be~$0$. If $\RandomBland(F,B,\sigma)$ returns a basis $B'$, then $B'$ is an optimal solution when all variables not in $F\cup B'$ are required to be~$0$. This holds for any permutation~$\sigma$. The permutation~$\sigma$ may determine, however, the sequence of pivoting steps performed. (In both cases, the proof follows from the fact that all pivoting steps performed reduce the value of the objective function, and that when the algorithm terminates, there are no such improving pivoting steps.)


\begin{figure}[t]
\begin{center}
\parbox{3in}{
\SetAlgoFuncName{Algorithm}{anautorefname}
\begin{function}[H]
\DontPrintSemicolon
\SetAlgoRefName{}
\eIf{$F =B$}
{\Return{$B$}\;}
{
    $i\gets \RANDOM(F\setminus B)$ \;
    $B' \gets \RandomFacet(F \setminus\{i\},B)$ \;
    \eIf{$\IMPROVE(B',i)$}
    {
        $B'' \gets \PIVOT(B',i)$ \;
        \Return{$\RandomFacet(F,B'')$}\;
    }
    {
      \Return{$B'$}\;
    }
}
\caption{\RandomFacet($F,B$)}
\end{function}
}
\\[10pt]
\parbox{3in}{
\SetAlgoFuncName{Algorithm}{anautorefname}
\begin{function}[H]
\DontPrintSemicolon
\SetAlgoRefName{}
\eIf{$F =B$}
{\Return{$B$}\;}
{
    $i\gets \argmin_{j \in F \setminus B} \sigma(j)$ \;
    $B' \gets \RandomPerm(F \setminus\{i\},B,\sigma)$ \;
    \eIf{$\IMPROVE(B',i)$}
    {
        $B'' \gets \PIVOT(B',i)$ \;
        \Return{$\RandomPerm(F,B'',\sigma)$}\;
    }
    {
      \Return{$B'$}\;
    }
}
\caption{\RandomPerm($F,B,\sigma$)}
\end{function}
}
\hspace*{10pt}
\parbox{3in}{
\SetAlgoFuncName{Algorithm}{anautorefname}
\begin{function}[H]
\DontPrintSemicolon
\SetAlgoRefName{}
\eIf{$F =\emptyset$}
{\Return{$B$}\;}
{
    $i\gets \argmin_{j \in F} \sigma(j)$ \;
    $B' \gets \Bland(F \setminus\{i\},B,\sigma)$ \;
    \eIf{$\IMPROVE(B',i)$}
    {
        $B'' \gets \PIVOT(B',i)$ \;
        \Return{$\Bland(F,B'',\sigma)$}\;
    }
    {
      \Return{$B'$}\;
    }
}
\caption{\Bland($F,B,\sigma$)}
\end{function}
}
\end{center}
\caption{The \RandomFacet, $\RandomPerm$, and $\Bland$ algorithms.}\label{F-RandomFacet}
\end{figure}

In~\cite{FriedmannHansenZwick/SODA11}, we erroneously claimed that for every linear program, the expected number of pivoting steps performed by \RandomFacet\ and \RandomPerm, with a randomly chosen permutation~$\sigma$, is the same. A counterexample to this claim is given in~\cite{FrHaZw14}.
The `proof' given in~\cite{FriedmannHansenZwick/SODA11} relied on the linearity of expectations, claiming that the fact that the two recursive calls of \RandomPerm\ share some random choices does not effect the expected number of pivoting steps performed. The `proof', however, contained a subtle flaw. Thus, while \RandomFacet\ has a subexponential upper bound on its complexity, no such subexponential upper bounds are currently known for \RandomPerm\ and \RandomBland.

We next describe an equivalent non-recursive version of \RandomFacet. Let $\xx$ be the current bfs, corresponding to the set~$B$ of basic indices, and the set~$N$ of non-basic indices. The algorithm maintains a \emph{permutation} $\langle i_1,i_2,\ldots,i_{d}\rangle$, where $d=n-m$, of the indices of~$N$. If~$\xx$ is the initial bfs, then a random permutation of~$N$ is chosen. If $\bcc\ge \0$, then $\xx$ is optimal, and we are done. Otherwise, choose the \emph{first} index~$i_j$ in the permutation for which $\bcc_{i_j}<0$. Perform a pivoting step from~$B$ to $B'=B\cup\{i_j\}\setminus\{i'\}$, for some $i'\not\in B\cup\{i_j\}$, so that  $N'=N\setminus\{i_j\}\cup\{i'\}$. The permutation corresponding to~$N'$ is obtained by randomly permuting $i_1,\ldots,i_{j-1},i'$, and keeping the original order of $i_{j+1},\ldots,i_d$. Proving, by induction, the equivalence of the recursive and non-recursive definitions of \RandomFacet\ is an instructive exercise. The permutation of~$N$ kept by the algorithm is exactly the `memory' referred to earlier. It also corresponds to the stack kept by the recursive version of \RandomFacet.

The non-recursive formulation of \RandomFacet\ is similar, yet different, from the following randomized version of Bland's rule. Choose a random permutation $\langle i_1,i_2,\ldots,i_{n}\rangle$ of $[n]$. At each step, choose the \emph{first} index~$i_j$ in the permutation for which $\bcc_{i_j}<0$ to enter the basis. We show below that this randomized version of Bland's rule is exactly the non-recursive version of the variant \Bland\ given above. There are two main differences between the non-recursive versions of \RandomFacet\ and \Bland. The first is that \Bland\ uses a permutation of the indices of all variables, not only those that are currently non-basic. The second is that \Bland\ uses the same permutation throughout the operation of the algorithm, while \RandomFacet\ `refreshes' the permutation at each step by randomly permuting the prefix $i_1,\ldots,i_{j-1},i'$.

The proof that the recursive and non-recursive formulations of \Bland\ are equivalent follows easily by induction. To get the exact equivalence, we assume that $\Bland$ chooses the \emph{last} index~$i_j$ in the permutation~$\sigma$ for which $\bcc_{i_j}<0$ to enter the basis. Let $\sigma^{-1}=\langle i_1,i_2,\ldots,i_n \rangle$. Note that every recursive call of $\Bland$ is of the form $\Bland(F_k,B,\sigma)$, where $F_k=\{i_k,i_{k+1},\ldots,i_n\}$ for some $k\in[n+1]$. (We let $F_{n+1}=\emptyset$.) The non-recursive algorithm iteratively considers $i_n,i_{n-1},\ldots$ until finding the first index~$j$ for which $\IMPROVE(B,j)$ returns true. If no such improving switch is found, then~$B$ is an optimal basis, and the algorithm terminates. We let $\Bland'(k,B,\sigma)$ denote the non-recursive algorithm which only considers the indices $i_n,i_{n-1},\ldots,i_k$. We claim that for every linear program, every initial basis $B\subseteq [n]$, every permutation $\sigma$, and every $k\in[n+1]$, $\Bland(F_k,B,\sigma)$ and $\Bland'(k,B,\sigma)$ perform exactly the same sequence of pivoting steps and hence return the same basis. Suppose that the claim is true for all larger values of~$k$, and for all better bases $B$. (A basis $B'$ is better than $B$ if the value of the bfs corresponding to~$B'$ is smaller than that of~$B$.) If~$B$ is an optimal basis, then both algorithms stop immediately. Otherwise, \Bland\ performs the recursive call $\Bland(F_{k+1},B,\sigma)$. By induction, this recursive call is equivalent to $\Bland(k+1,B,\sigma)$, so both return the same basis $B'$. Now, if $\IMPROVE(B',i_k)$ is false, then both $\Bland(F_k,B,\sigma)$ and $\Bland'(k,B,\sigma)$ are done. Otherwise, they both perform the pivoting step $B''\gets \PIVOT(B',i_k)$ and then perform $\Bland(F_k,B'',\sigma)$ and $\Bland'(k,B'',\sigma)$, respectively. As~$B''$ is a better basis than $B$, it follows by induction that these two calls are again equivalent.

\section{Shortest paths}\label{S-short-path}

In~\cite{FriedmannHansenZwick/SODA11,FriedmannHansenZwick/STOC11} we showed that \RandomPerm\ performs a subexponential number of pivoting steps on some linear programs that correspond to MDPs. Here we obtain similar subexponential lower bounds also for \RandomFacet\ and \RandomBland. Furthermore, we show that such lower bounds for \RandomFacet, \RandomPerm\ and \RandomBland\ can be obtained using linear programs that correspond to the \emph{purely combinatorial} problem of finding \emph{shortest paths} in directed graphs. The graphs we use are even acyclic and all their edge weights are non-negative. It is slightly more convenient for us to consider shortest paths from all vertices \emph{to} a given \emph{target} vertex, rather than shortest paths \emph{from} a given \emph{source} vertex. The two problems, however, are clearly equivalent.

Let $G=(V,E,c)$ be a weighted directed graph, where $c:E\to\RR$ is a \emph{cost} (or \emph{length} function) defined on its edges. Let $\TT\in V$ be a specific vertex designated as the \emph{target} vertex. We let $n=|V\setminus\{\TT\}|$ and $m=|E|$ be the number of vertices, not counting the target, and edges in~$G$, respectively. We are interested in finding a tree of shortest paths from all vertices to~$\TT$. The problem, for general graphs with positive and negative edge weights, but no negative cycles, can be solved in $O(mn)$ time using a classical algorithm of Bellman and Ford \cite{Bellman58,Ford56}. When the edge weights are non-negative, the problem can be solved in $O(m+n\log n)$ time using Dijkstra's algorithm \cite{Di59}. When the graph is acyclic, the problem can be easily solved in $O(m+n)$ time.

The simplex algorithm, specialized to the \emph{min cost flow} problem, is usually referred to as the \emph{network simplex} algorithm. For a thorough treatment of the network simplex algorithm, see
Chv{\'a}tal \cite{Chvatal83}, Ahuja \etal\ \cite{AhMaOr93}, and Bertsimas and Tsitsiklis \cite{Bertsimas1997introduction}. As the shortest paths problem is a very special case of the min cost flow problem, it can also be solved using the network simplex algorithm.

The shortest path problem can be formulated as a min cost flow problem, and hence as a linear program, as follows. Finding a tree of shortest paths from all vertices to the target $\TT$ is equivalent to finding the min cost flow in which we have a supply of one unit at each vertex, other than~$\TT$, and a demand of $n$ units at~$\TT$. (Recall that~$n$ is the number of non-terminal vertices.) The corresponding primal and dual linear programs are:
\[
(P)~
\begin{array}{llll}
\min & \cc^T \xx&& \\
\mbox{s.t.} &\AAA \xx &=& \ee \\
     &\xx&\ge& \0
\end{array} \quad\quad\quad\quad
(D)~
\begin{array}{llll}
\max & \ee^T \yy&& \\
\mbox{s.t.} &\AAA^T \yy &\le& \cc \\
\end{array}
\]
where $\AAA\in \RR^{n\times m}$ is the \emph{incidence matrix} of the graph. We assume, without loss of generality, that $V\setminus\{\TT\}=[n]$.~\footnote{When considering linear programs in standard form, it is customary to use~$n$ for the number of variables and~$m$ the number of equality constraints. When considering graphs, it is customary to use~$n$ for the number of vertices and~$m$ for the number of edges. Unfortunately, these two conventions clash in our case. We switch now to graph terminology, so~$n$ is now the number of vertices, and hence the number of equality constraints, and~$m$ is the number of edges, and hence the number of variables and inequality constraints.} Each row of~$\AAA$ corresponds to a vertex of the graph~$G$. Each column of~$\AAA$ corresponds to an edge of~$G$. Each column contain a single $+1$ entry, and possibly a $-1$ entry. If the $i$-th edge is $(j,k)$, then $\AAA_{j,i}=+1$ and $\AAA_{k,i}=-1$. All other entries of the $i$-th column are zeros. If the $i$-th edge is $(j,\TT)$, then $\AAA_{j,i}=+1$ and all other entries of the $i$-th column are zeros. (Note that some authors define the incidence matrix to be the negation of our incidence matrix.) The vector $\cc\in\RR^m$ contains the edge costs. The vector $\ee\in \RR^n$ is the all one vector. The primal variable vector $\xx\in\RR^m$ is the \emph{flow} vector, specifying the flow on each edge. The constraints $\AAA\xx=\ee$ ensure that the net flow out of each vertex is $1$. Note that the target $\TT$ does not appear explicitly in this formulation. Finally, $\yy\in \RR^n$, the dual variables vector, directly specifies stipulated distances from each vertex to~$\TT$.

It is not difficult to check that every bfs~$B$ of $(P)$ corresponds to a \emph{tree} containing paths from all vertices to~$\TT$. If $i\in B$, i.e., the $i$-th edge $e_i=(j,k)$ belongs to the tree, then $\xx_i$, the flow on $e_i$, is the number of \emph{descendants} of~$j$ in this tree, including~$j$ itself. The cost of~$B$ is thus the sum of the lengths of the paths along the tree from all vertices to~$\TT$. The cost is therefore minimized when~$B$ corresponds to a shortest paths tree. The dual variables corresponding to~$B$ are the \emph{distances} along the tree. Thus, $\yy_j$ is simply the length of the path from~$j$ to~$\TT$ in the tree. The reduced costs~$\cc$ also have a simple combinatorial interpretation. If $e_i=(j,k)$, then $\bcc_i=\cc_i+\yy_k-\yy_j$. (If $e_i=(j,\TT)$, then $\bcc_i=\cc_i-\yy_j$.) Note that if $\bcc_i<0$, then $\cc_i+\yy_k<\yy_j$, and thus the path from~$j$ to~$\TT$ that starts with the edge $e_i=(j,k)$, which is currently not in the tree, is shorter than the path from~$j$ to~$\TT$ along the tree. The tree can therefore be improved by performing a \emph{switch} in which the edge $e_i=(j,k)$ is inserted, and the edge of the tree currently emanating from~$j$ is removed. If the graph does not contain \emph{negative} cycles, then a new and improved tree is obtained. This is exactly the pivoting step in which~$i$ enters the basis.

Each tree~$B$ of paths from all vertices to the terminal~$\TT$ is obtained by choosing one outgoing edge from each non-terminal vertex. This may also be viewed as specifying a \emph{policy} in a 1-player game.
In the case of shortest paths with no negative cycles, the set of edges chosen, i.e., the policy, is required to be \emph{acyclic}. In more general classes of games, a policy may contain cycles.

In the sequel, we view a basis, i.e., tree, $B$, not as a set of indices but as a set of edges. We also let $y_B(j)=\yy_j$ denote the distance in~$B$ from vertex~$j$ to~$\TT$, for every $j\in [n]$.
If $e=(i,j)$ is an edge of the
graph and $c(i,j)+y_B(j)<y_B(i)$, then~$e$ is said to be an
\emph{improving switch} with respect to~$B$. We let $B[e]$ be
the tree obtained by performing the switch, i.e.,
$B[e]=B\cup\{e\}\setminus\{e'\}$, where~$e'$ is the edge of~$B$
emanating from~$i$. We let $y(j)$ denote the \emph{distance} from~$j$ to~$\TT$ in the graph, i.e., $y(j)=y_{B*}(j)$, where $B^*$ is a shortest paths tree. An edge $e=(i,j)$ is said to be \emph{optimal} if it appears on some shortest path from~$i$ to~$\TT$. It is not difficult to see that an edge $(i,j)$ is optimal if and only if $y(i)=c(i,j)+y(j)$, and
a tree~$B$ is optimal, i.e., is a tree of shortest paths, if and only if all its edges are optimal.

We end the section by describing the way \RandomFacet\ is used to find a tree of shortest paths in a weighted directed graph $G=(V,E,c)$ to a terminal~$\TT$ using the terminology used throughout the rest of the paper. The algorithm starts with an initial tree $B$. It randomly chooses an edge~$e$ \emph{not} in~$B$. A recursive call of the algorithm is used to find a tree~$B'$ of shortest paths in the graph $G\setminus\{e\}$ obtained by removing~$e$. If~$e$ is not an improving switch with respect to~$B'$, then $B'$ is also a tree of shortest paths of the original graph~$G$. Otherwise, the improving switch $B''\gets B'[e]$ is performed, and the algorithm is called recursively on the whole graph with~$B''$.

%% file: Simplified-Construction.tex
\section{A high-level description of the lower bound proofs}\label{S-randomized-counter}

Our lower bounds for the $\RandomFacet$, $\RandomPerm$, and
$\RandomBland$ algorithms are obtained by simulating the behavior of
the \emph{randomized counter} shown in Figure~\ref{fig: randomized
  counter}.
Such a counter is composed of $n$ bits, $bit_1,\ldots,bit_n$, all
initially~$0$.
The randomized counter works in a
recursive manner, focusing each time on a subset $N\subseteq [n]:=
\{1,\ldots,n\}$ of the bits, such that $bit_i =
0$ for all $i \in N$. Initially $N=[n]$. If $N=\emptyset$,
then nothing is done. Otherwise, the counter chooses a random index
$i\in N$ and recursively performs a randomized count on
$N\setminus\{i\}$. When this recursive count is done, we have
$bit_j = 1$, for every $j\in N\setminus\{i\}$, while
$bit_i=0$. Next, the $i$-th bit is set to~$1$, and all bits
$j \in N \cap [i-1]$ are reset to~$0$.
Finally, a recursive randomized count is performed on $N\cap[i-1]$.

\begin{figure}[t]
\begin{center}

\parbox{3in}{
\begin{function}[H]
\DontPrintSemicolon
\If{$N\neq\emptyset$}
{
    $i \gets \RANDOM(N)$ \;
    $\RandCount(N \setminus \{i\})$ \;
    $bit_i \gets 1$ \;
    \lFor{$j \in N \cap [i-1]$}{~$bit_j \gets 0$} \;
    $\RandCount(N \cap [i-1])$ \;
}
\caption{\RandCount($N$)}
\end{function}
}
\hspace*{10pt}
\parbox{3in}{
\begin{function}[H]
\DontPrintSemicolon
\If{$N\neq\emptyset$}
{
    $i \gets \argmin_{j \in N}\; \sigma(j)$\;
    $\RandCount^{1P}(N \setminus \{i\}, \sigma)$ \;
    $bit_i \gets 1$ \;
    \lFor{$j \in N \cap [i-1]$}{~$bit_j \gets 0$} \;
    $\RandCount^{1P}(N \cap [i-1], \sigma)$ \;
}
\caption{\RandCount$^{1P}$($N,\sigma$)}
\end{function}
}
\end{center}
\caption{The randomized counter. Original version on the left. One-permutation variant on the right.}\label{fig: randomized counter}
\end{figure}

Let $f(n)$ be the expected number of times the call $\RandCount([n])$
sets a bit of the randomized counter to~$1$. It is not difficult to
check that the behavior of $\RandCount(N)$, where $|N|=k$, is
equivalent to the behavior of $\RandCount([k])$. It is then easy to see
that $f(n)$ satisfies the following recurrence relation:
\begin{align*}
f(0) ~&=~ 0 \\
f(n) ~&=~ f(n-1) + 1 + \frac{1}{n} \sum_{i=0}^{n-1} f(i)
\quad\text{for $n > 0$}
\end{align*}

\begin{lemma}\label{L-recurrence}
$\displaystyle f(n) \;=\; \sum_{k=1}^n \frac{1}{k!} {n \choose k} ~.$
\end{lemma}

The proof of the lemma can be found in Appendix~\ref{A-recurrence}.

According to Lemma~\ref{L-recurrence} we can interpret
$f(n)$ as the expected number of increasing subsequences in a
uniformly random permutation of $[n] = \{1,\dots,n\}$, i.e., every non-empty
subset $S \subseteq [n]$ appears as an increasing subsequence with
probability $1/|S|!$.
The asymptotic behavior of~$f(n)$ is known quite precisely:
\begin{lemma}[\cite{LiPi81},\mbox{\cite[p. 596--597]{flajoletsedgewick/analyticcombinatorics}}] \label{lemma:asymp}
$\displaystyle
f(n) \sim \frac{{\rm e}^{2\sqrt{n}}}{
  2\sqrt{\pi{\rm e}}\, n^{1/4}} 
$
\end{lemma}
Note, in particular, that $f(n)$ is subexponential. The challenge,
of course, is to construct weighted directed graphs such that
the behavior of $\RandCount$ is mimicked by the
$\RandomFacet$, $\RandomPerm$, and
$\RandomBland$ algorithms.

Just like it was natural to define one-permutation variants of the
$\RandomFacet$ algorithm, the randomized counter, $\RandCount$, can be
implemented such that the random choices are based on a single, given,
random permutation $\sigma: [n] \to [n]$. We refer to the
one-permutation variant of $\RandCount$ as $\RandCount^{1P}$. It
is also shown in Figure \ref{fig: randomized counter}. We let
$f^{1P}(N,\sigma)$, where $N \subseteq [n]$, be the number of times
the call $\RandCount^{1P}(N,\sigma)$ sets a bit to 1. Note that
$\RandCount^{1P}$ is deterministic, and that:
\begin{align*}
f^{1P}(\emptyset,\sigma) &~=~ 0 \\
f^{1P}(N,\sigma) &~=~ f^{1P}(N\setminus\{i\},\sigma) + 1 + f^{1P}(N
\cap [i-1],\sigma) \quad\quad\text{where~ $N \ne \emptyset$~ and ~$i \in
  \argmin_{j \in N}\; \sigma(j)$}
\end{align*}
We let $f^{1P}(n)$ be the expected value of $f^{1P}([n],\sigma)$ when
$\sigma$ is a uniformly random permutation of $[n]$.
The following lemma is proved by using induction and linearity of
expectation. The proof of the lemma can be found in Appendix
\ref{A-counters}. \footnote{The fact that the expected number of steps performed by $\RandCount$ and  
$\RandCount^{1P}$ is the same led us to (mistakenly) believe that the expected number of steps performed by \RandomFacet\ and \RandomPerm\ is also the same.}
\begin{lemma}\label{lemma:counters}
$f(n) = f^{1P}(n)$.
\end{lemma}

\subsubsection*{Lower bound graphs and interpretation of trees.}
We next
describe how to obtain the graphs used for our lower bounds. We use the same family
of graphs for all three lower bounds. The graphs $G_{n,r,s,t}$ are
parameterized by four parameters, $n,r,s,t \in \NN$, where $n$ is the
number of bits. For simplicity
we initially consider the case when $r=s=t=1$. The
graphs are acyclic and are composed of~$n$ levels. The $i$-th level represents the $i$-th bit of the counter, with the $n$-th bit being the
most significant. Figure \ref{fig:one_level_simple} shows the $i$-th level of
the graph. The target~$\TT$ is identified with~$u_{n+1}$ and~$w_{n+1}$. The initial tree consists entirely of edges with non-zero cost.

\begin{figure}[t]
\center
\includegraphics[scale=0.8]{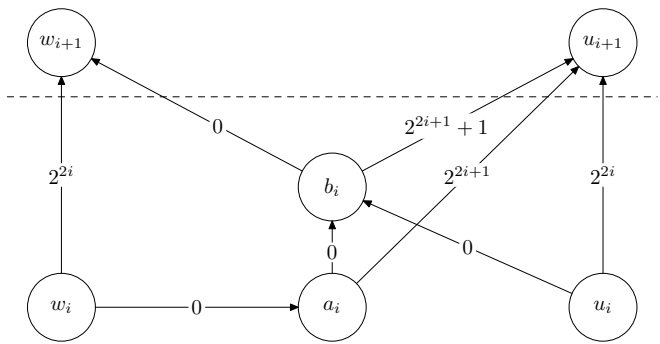}
\caption{The $i$-th level of the graph $G_{n,r,s,t}$ when $r=s=t=1$.}
\label{fig:one_level_simple}
\end{figure}

The critical vertex is the vertex $b_i$. A tree~$B$
consists of one out-going edge from every vertex. Since the choice at
$b_i$ is binary, we can interpret~$B$ as a setting of the binary
counter as follows. If $(b_i,w_{i+1}) \in B$ we say that the $i$-th
bit is 1 and write $bit_i(B) = 1$. Otherwise we say that the $i$-th
bit is 0 and write $bit_i(B) = 0$.

The critical edge is the edge $(b_i,w_{i+1})$. Recall that the
$\RandomFacet$, $\RandomPerm$, and $\RandomBland$ algorithms work by
removing edges and recursively solving sub-problems. When
$(b_i,w_{i+1})$ is removed, the $i$-th bit is fixed to 0 during the
following recursive call, which exactly corresponds to the behavior of
$\RandCount$. In order to argue that the recursive call sets the
remaining bits to 1, we characterize optimal trees when certain key
edges, such as $(b_i,w_{i+1})$, are removed.

\subsubsection*{Optimal trees when edges are removed.}
We say that level $i$ is \emph{stable} in a tree $B$ if
either $B$ contains all the~0 cost edges of the $i$-th level, or
$B$ contains all the non-zero cost edges of the $i$-th
level. We also say that $B$ is stable from level $i$ if the $j$-th level is
stable in $B$ for all $j \ge i$. It is not difficult to check that
$y_B(u_i) = y_B(w_i)$ if $B$ is stable from level $i$. Moreover,
if $B$ is stable from $i+1$ and no edge is missing in the $i$-th
level, then it is optimal to use the 0 cost edges of level~$i$. On the other hand, if
$(b_i,w_{i+1})$ is missing, then it is optimal to use the non-zero cost edges. In particular, the distances decrease as more and more
bits are set to 1.

It follows from the above discussion that when $(b_i,w_{i+1})$ is
removed the optimal tree is stable for all levels, and all bits,
except $i$, are 1. After returning from the first recursive call, the
$\RandomFacet$, $\RandomPerm$, and $\RandomBland$ algorithms
all perform the improving switch $(b_i,w_{i+1})$, setting the
$i$-th bit to 1 and making the level unstable. Suppose
the edge $(a_i,b_i)$ is removed next, so that no path from~$w_i$ to~$w_{i+1}$ has cost~0. The distance from~$u_i$ to~$w_{i+1}$ remains~0, and for the optimal tree $B$ we get $y_B(w_i) =
2^{2i} + y_B(u_i)$. In level $i-1$ all vertices then go to $u_{i}$, so that
$y_B(w_{i-1}) = 2^{2(i-1)} + y_B(u_{i-1})$, and by induction
all the lower bits are reset. Performing the improving switch $(a_i,b_i)$ reduces the cost of reaching
$w_{i+1}$ from $w_i$ to 0, which enables the lower bits to count again. Note that the lower levels are unstable at this stage, but it is, in fact, only important that
$(b_j,w_{j+1})\not\in B$ for $j < i$.

\subsubsection*{Gadgets and full construction.} The behavior described
above requires $\RandomFacet$, $\RandomPerm$, and
$\RandomBland$ to pick $(b_i,w_{i+1})$ before $(a_i,b_i)$, for $i
\in [n]$, and to pick the remaining
edges after $(a_i,b_i)$. We introduce two basic
gadgets to ensure this order occurs with high probability.

The first idea is to duplicate edges whose removal should be
delayed. We thus make $t$ copies of every edge other than
$(b_i,w_{i+1})$ and $(a_i,b_i)$, for $i \in [n]$. The second idea is to speed up the removal of an edge by replacing it with a path; removing
a single edge along the path corresponds to removing the original
edge. We thus replace $(b_i,w_{i+1})$ and $(a_i,b_i)$, for
$i \in [n]$, by paths of length $s$. We show that $(b_i,w_{i+1})$ and $(a_i,b_i)$, for $i\in [n]$, are removed before other edges with high probability when $s$ and $t$ are chosen appropriately. We apply the same
idea to ensure that $(b_i,w_{i+1})$ is
removed before $(a_i,b_i)$. We make $r$ copies of the path corresponding to $(a_i,b_i)$, and
we make the $(b_i,w_{i+1})$-path $r$ times longer.

The resulting graph $G_{n,r,s,t}$ is shown in Figure
\ref{fig:one_level}. Every vertex on
an edge-path can escape the path with an edge that corresponds to the
original choice. In order for the paths to be reset correctly,
the costs along the paths increase in increments
of size $\epsilon = \frac{1}{rs}$, making it
cheaper to escape a path as soon as possible.

\begin{figure}[t]
\center
\includegraphics[scale=0.8]{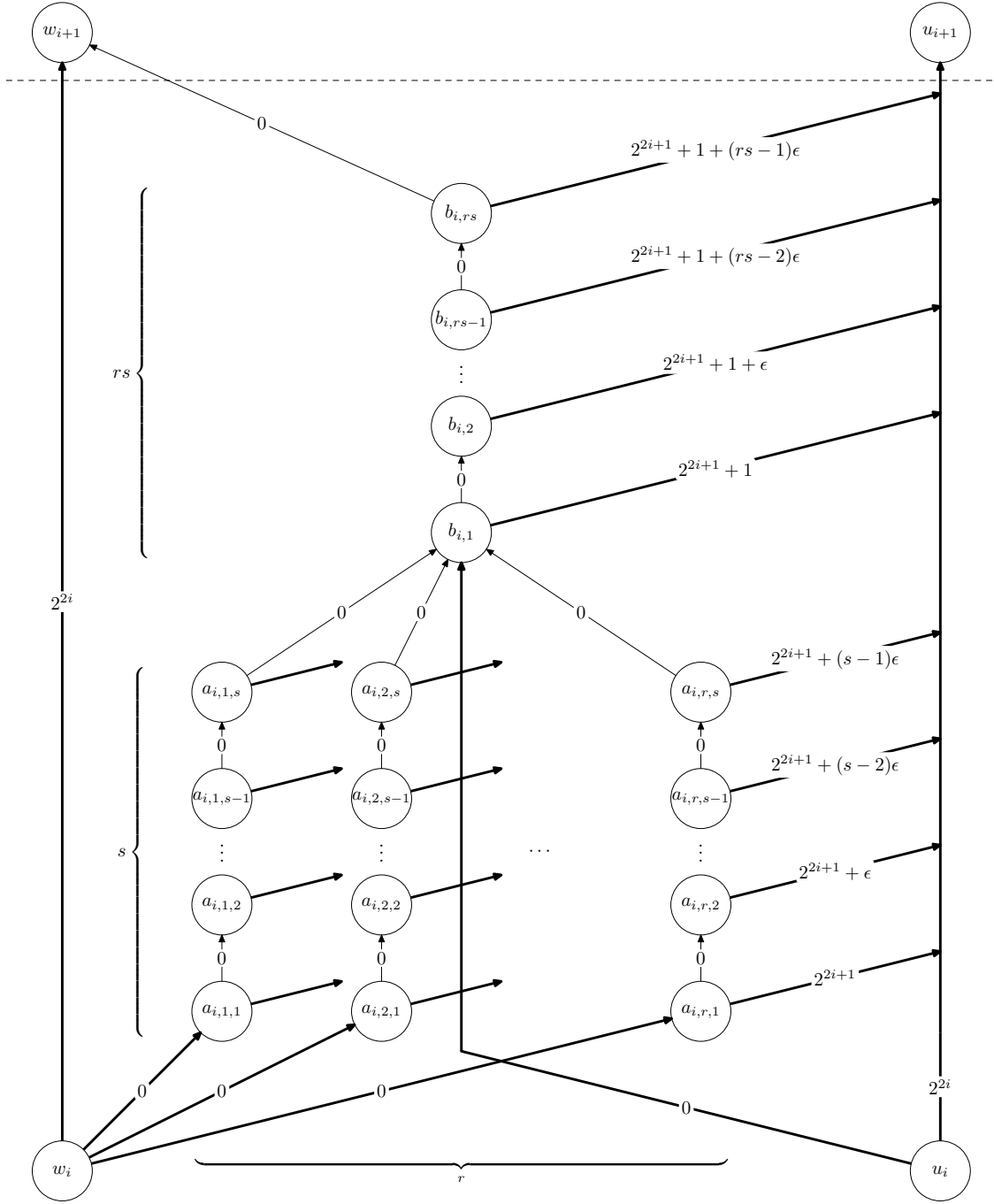}
\caption{The $i$-th level of the graph $G_{n,r,s,t}$. The bold edges
  are multi-edges with multiplicity $t$.}
\label{fig:one_level}
\end{figure}

\subsubsection*{Lower bound for $\RandomPerm$.} Recall that $\RandomPerm$ takes as input a set of edges $F$, a tree~$B$, and a permutation $\sigma$, and that it initially removes the first edge $e \in F \setminus B$ according to $\sigma$. We interpret~$F$ and~$B$ as a configuration of the randomized counter as follows. For every $i \in [n]$, let $\bbb^1_i$ be the set of edges along the $(b_i,w_{i+1})$-path. If $\bbb^1_i \not\subseteq F$ then the $i$-th bit is fixed to~0. If $\bbb^1_i \subseteq F$ and $\bbb^1_i \cap B = \emptyset$ then the $i$-th bit is 0, and otherwise it is 1.

Note that $B$ need not contain all of $\bbb^1_i$ for the $i$-th bit to be 1. When $i$ is fixed to 0 by removing $e \in \bbb^1_i$, the resulting optimal tree $B'$ contains the edges in $\bbb^1_i$ ahead of $e$ but not the edges behind $e$. In \cite{FriedmannHansenZwick/SODA11} and \cite{FriedmannHansenZwick/STOC11} we used stronger gadgets, that cannot be implemented for shortest paths, to ensure that $B'$ included $\bbb^1_i \setminus \{e\}$. We observe, however, that the current tree is less important than the set of remaining edges, and this allows us to use the path gadget.

Recall that $\RandCountP$ is a one-permutation variant of the randomized counter, and that it fixes the first available bit according to a given permutation $\hat\sigma: [n] \to [n]$. Recall also that, according to Lemma \ref{lemma:counters}, $\RandCountP$ and $\RandCount$ perform the same expected number of increments when $\hat\sigma$ is uniformly random.
The permutation $\sigma$ of the edges defines an induced permutation~$\hat\sigma$ of the bits, where $\hat{\sigma}$ is obtained from $\sigma$ from the first element of $\bbb^1_i$ for every $i \in [n]$. We show that $\RandomPerm(F,B,\sigma)$ performs at least as many improving swithces as $\RandCountP(N(F,B),\hat\sigma)$, where $N(F,B)$ is the set of unfixed 0 bits. To be precise, we assume that $\sigma$ is \emph{well-behaved} and that $B$ has the relevant structure.

A permutation is well-behaved if it ensures that, for all $i\in [n]$, an edge from~$\bbb^1_i$ is removed before $w_i$ is disconnected from $b_{i,1}$, which in turn happens before any multi-edge is exhausted. Concretely, at least one edge from $\bbb^1_i$ comes before all edges from one of the $r$ copies of the $(a_i,b_i)$-path. Also, at least one edge from each $(a_i,b_i)$-path comes before the last copy of every multi-edge.

The proof is by induction, and the critical case is when an improving switch increments a bit $i$, i.e., when $\bbb^1_i \subseteq F$, $\bbb^1_i \cap B = \emptyset$, and the chosen edge~$e$ is from $\bbb^1_i$. In fact these are the only improving switches we count in the analysis, and it is the only situation for which the count includes both recursive calls. For all other situations we specify a single recursive call that is considered for the induction step:
\begin{enumerate}
\item
When $e \in \bbb^1_i$, $\bbb^1_i \subseteq F$, and $\bbb^1_i \cap B \ne \emptyset$, the count is for the second call.
\item
When $e$ belongs to an $(a_i,b_i)$-path, $\bbb^1_i \subseteq F$, and removing $e$ disconnects $w_i$ from $b_{i,1}$, the count is for the second call.
\item
In all other cases the count is for first call.
\end{enumerate}

Note that the structure of the graph is preserved by the second recursive call since no edge is removed. Also, due to the assumption that the permutation $\sigma$ is well-behaved, only redundant edges are removed in the third case.

Let us note that the first case is new compared to our analysis in \cite{FriedmannHansenZwick/SODA11} and \cite{FriedmannHansenZwick/STOC11}. It is needed because not all edges from $\bbb^1_i$ are immediately included in a tree where the $i$-th bit is set to 1. We observe, however, that during the first recursive call the tree is only updated by including edges along the $(b_i,w_{i+1})$-path ahead of the chosen edge~$e$. The remaining edges are included when the tree is updated in a situation corresponding to the second case, which also causes the lower bits to reset. A similar issue occurs after a reset.

Finally, we show that a uniformly random permutation is well-behaved with high probability when $r=s=t=3 \lceil\log n\rceil$. Since a uniformly random well-behaved permutation gives a uniformly random induced permutation, this shows that $\RandomPerm$ simulates $\RandCountP$ with high probability. Thus, when the total number of edges is $M = \tilde\Theta(n)$, the expected number of improving switches performed is $2^{\tilde\Omega(\sqrt{n})}$, which proves a $2^{\tilde\Omega(\sqrt{M})}$ lower bound.

\subsubsection*{Lower bound for $\RandomBland$.} The proof for $\RandomBland$ is similar to the proof for $\RandomPerm$. Recall that $\Bland(F,B,\sigma)$ can remove edges from $F$ that are in $B$, meaning that the invariant $B \subseteq F$ is not maintained. The only guarantee is that the tree $B^\dagger$ returned by the algorithm is optimal for the subgraph defined by $F \cup B^\dagger$. The edges in $B \setminus (F \cup B^\dagger)$ are lost and can only be reintroduced higher up in the recursion.

To use the same approach as for $\RandomPerm$ we keep track of the edges in $B$ that remain in $B^\dagger$. For this purpose we introduce the notion of \emph{fixed edges}, which are edges that are optimal for $F\cup B$. Note that, since the graph is acyclic, if an edge $e$ is removed and $B'$ is optimal for the resulting graph, then $B'$ remains optimal at every vertex ahead of $e$ when $e$ is reintroduced. In particular, when the $i$-th bit is set to 1, the edges at higher levels become fixed and maintain their configuration while the lower bits are reset and used for additional counting. This shows that the desired edges remain in $B^\dagger$.

Another challenge is to show that edges that should be removed do not remain in $B^\dagger$. Suppose, for instance, that the $i$-th bit was set to 1, and that the lower bits are being reset. If the algorithm picks an edge $e \in \bbb^1_j$ with $\bbb^1_j \subseteq F \cap B$ and $j < i$, then the $j$-th bit should be fixed to 0, but since $e \in B$ this may not be the case. We show however that, when the permutation $\sigma$ is well-behaved, the lower bits are reset before any additional counting is done, such that the edge $e$ is truly removed.

The remainder of the proof is the same as for $\RandomPerm$, and we again set $r=s=t=3 \lceil\log n\rceil$ to get a $2^{\tilde\Omega(\sqrt{M})}$ lower bound, where $M = \tilde\Theta(n)$ is the number of edges.

\subsubsection*{Lower bound for $\RandomFacet$.} The proof of the lower bound for $\RandomFacet$ is more involved than the proofs for $\RandomPerm$ and $\RandomBland$. The main complication is that nothing prevents multi-edges from being exhausted. For $\RandomPerm$ and $\RandomBland$, as long as one instance of a multi-edge appears sufficiently late in the given permutation, then that multi-edge is not exhausted until after the desired counting behavior is observed. For $\RandomFacet$, on the other hand, the multi-edges are always available for removal, and the edges along $(b_i,w_{i+1})$-paths and $(a_i,b_i)$-paths are often part of the current tree and cannot be removed.

We use a technique by G{\"a}rtner~\cite{Gartner02} to bound the expected number of improving switches by bounding the probabilities that certain \emph{computation paths} are generated by the algorithm. The call $\RandomFacet(F,B)$ generates a random, binary computation tree where the nodes of the tree correspond to recursive calls. The tree is defined by the edges from the graph that are picked at the beginning of every recursive call. A computation path is a sequence of successive recursive calls that form a path from the root in a computation tree. It is described by the sequence of chosen edges, and by a sequence of directions that indicate whether the path follows the branch of the first or the second recursive call.

The edges from the graph chosen along a computation path essentially represent a permutation, which allows us to use the same type of arguments as in the proofs for $\RandomPerm$ and $\RandomBland$. We restrict our attention to \emph{canonical paths} that correspond to well-behaved permutations. We also restrict our attention to computation paths where at most $\sqrt{n}$ bits are set to 1 by second recursive calls along the path. The probability of generating such shorter paths dominate the probability of generating longer paths. We use a Chernoff bound argument to bound the probability that a multi-edge is exhausted along a computation path, and the fact that we restrict our attention to shorter paths allow us to obtain a better bound.

We show that when $r = \Theta(\log n)$, $s = \Theta(\sqrt{n}\log n)$, and $t = \Theta(\log n)$, the expected number of improving switches performed by the $\RandomFacet$ algorithm is at least $2^{\Omega(\sqrt{n})}$. Thus, the number of edges is $M = \tilde\Theta(n^{3/2})$, which gives a lower bound of $2^{\tilde\Omega(\sqrt[3]{M})}$. Note that we increase the parameter $s$ instead of the parameter $t$ because adding more copies of an edge also makes those copies more likely to be removed.

%% file: Actual-Construction.tex
\section{The lower bound construction}\label{sec:construction}

In this section we formally define the family of weighted directed graphs
$G_{n,r,s,t}$, where $n,r,s,t \in \NN$, that we use to prove
lower bounds for $\RandomFacet$, $\RandomPerm$, and
$\RandomBland$. The graph $G_{n,r,s,t}$ can be divided into $n$
levels, where the $i$-th level corresponds to the $i$-th bit of an
$n$-bit binary counter. Figure \ref{fig:one_level} gives a schematic
description of one level of the graph. The graph contains one
additional \emph{target} vertex $\TT$ which we identify with
$u_{n+1}$ and $w_{n+1}$. We are interested in finding the shortest
paths from all vertices to $\TT$.

Formally, the graph $G_{n,r,s,t} = (V, E, c)$, where $c: E \to \RR$,
has vertex set
\[
V ~=~ \{\TT\} \cup \{u_i,w_i \mid i \in [n]\} \cup 
\{a_{i,j,k} \mid i \in [n], j\in [r], k \in [s]\}
\cup \{b_{i,j} \mid i \in [n], j\in [rs]\}\;,
\]
and the edges are specified in Table~\ref{figure: edges and
  priorities}. We use $\epsilon = 1/(rs)$ to define the costs. We
also assign a name to every edge. Note that the graph is acyclic, that
all costs are non-negative, and that there is an edge with cost
0 that leaves every vertex. Hence, any tree composed of edges
whose costs are all 0 is optimal.

\begin{table}[h]
\begin{center}
\center
\renewcommand{\arraystretch}{1.3}
\begin{tabular}[ht]{|c|rcl|c|c|c|}
\hline
  Quantification & \multicolumn{3}{c|}{Edge} & Cost & Name\\
  \hline
  \hline
  \multirow{2}{*}{$i \in [n], j\in [r], k\in[s-1], \ell \in [t]$} 
  & $a_{i,j,k}$ & $\to$ & $a_{i,j,k+1}$ & 0 & $a^1_{i,j,k}$\\
  & $a_{i,j,k}$ & $\to$ & $u_{i+1}$    & $2^{2i+1}+(k-1)\epsilon$ & $a^{0,\ell}_{i,j,k}$  \\
  \hline
  \multirow{2}{*}{$i \in [n], j\in [r], \ell \in [t]$} 
  & $a_{i,j,s}$ & $\to$ & $b_{i,1}$ & 0 & $a^1_{i,j,s}$\\
  & $a_{i,j,s}$ & $\to$ & $u_{i+1}$ & $2^{2i+1}+(s-1)\epsilon$ & $a^{0,\ell}_{i,j,s}$  \\
  \hline
  \multirow{2}{*}{$i \in [n], j\in [rs-1], \ell \in [t]$} 
  & $b_{i,j}$ & $\to$ & $b_{i,j+1}$ & 0 & $b^1_{i,j}$\\
  & $b_{i,j}$ & $\to$ & $u_{i+1}$    & $2^{2i+1}+1+(j-1)\epsilon$ & $b^{0,\ell}_{i,j}$  \\
  \hline
  \multirow{2}{*}{$i \in [n], \ell \in [t]$} 
  & $b_{i,rs}$ & $\to$ & $w_{i+1}$ & 0 & $b^1_{i,rs}$\\
  & $b_{i,rs}$ & $\to$ & $u_{i+1}$ & $2^{2i+1}+1+(rs-1)\epsilon$ & $b^{0,\ell}_{i,rs}$  \\
  \hline
  \multirow{2}{*}{$i \in [n], \ell \in [t]$} 
  & $u_{i}$ & $\to$ & $b_{i,1}$ & 0 & $u^{1,\ell}_{i}$\\
  & $u_{i}$ & $\to$ & $u_{i+1}$ & $2^{2i}$ & $u^{0,\ell}_{i}$  \\
  \hline
  \multirow{2}{*}{$i \in [n], j \in [r], \ell \in [t]$} 
  & $w_{i}$ & $\to$ & $a_{i,j,1}$ & 0 & $w^{j,\ell}_{i}$\\
  & $w_{i}$ & $\to$ & $w_{i+1}$ & $2^{2i}$ & $w^{0,\ell}_{i}$  \\
  \hline
\end{tabular}
\end{center}
\caption{Edges, costs, and names.}
\label{figure: edges and priorities}
\end{table}

It will be useful to partition the set of edges into smaller sets of
edges with similar roles. We define:
\[
\renewcommand{\arraystretch}{1.3}
\begin{array}{rrcl}
\forall i \in [n], \forall j \in [r]: &\quad \aaa^1_{i,j} &=&
\{a^1_{i,j,k} \mid k \in [s]\} \\
\forall i \in [n], \forall j \in [r], \forall k \in [s]: &\quad
\aaa^0_{i,j,k} &=& \{a^{0,\ell}_{i,j,k} \mid \ell \in [t]\}\\
\forall i \in [n]: &\quad \bbb^1_{i} &=& \{b^1_{i,j} \mid j \in [rs]\}
\\
\forall i \in [n], \forall j \in [rs]: &\quad \bbb^0_{i,j} &=& \{b^{0,\ell}_{i,j} \mid \ell \in [t]\}\\
\forall i \in [n]: &\quad \uuu^1_{i} &=& \{u^{1,\ell}_{i} \mid \ell \in [t]\}\\
\forall i \in [n]: &\quad \uuu^0_{i} &=& \{u^{0,\ell}_{i} \mid \ell \in [t]\}\\
\forall i \in [n], \forall j \in [r]: &\quad \www^j_{i} &=& \{w^{j,\ell}_{i} \mid \ell \in [t]\}\\
\forall i \in [n]: &\quad \www^0_{i} &=& \{w^{0,\ell}_{i} \mid \ell \in [t]\}
\end{array}
\]
Note that six of the definitions correspond to multi-edges with
multiplicity $t$. It will also be helpful to define a set of
multi-edges $\mathcal{M}$, i.e., $\mathcal{M}$ is a set of sets of edges:
\begin{align*}
\mathcal{M} ~=~ &\{\uuu^1_{i}, \uuu^0_{i}, \www^0_{i} \mid i \in [n]\} \,\cup\,
\{\bbb^0_{i,j} \mid i\in [n], j\in [rs]\} \,\cup\\
&\{\aaa^0_{i,j,k} \mid i\in [n], j \in [r], k \in [s]\} \,\cup\,
\{\www^j_{i}\mid i\in [n], j\in [r]\}\;.
\end{align*}

We are going to interpret certain trees $B \subseteq E$ as
configurations of the binary counter.
Note that the names of the edges are given superscripts, typically 0
or 1. We refer to edges with superscript 0 as zero-edges, and to the
remaining edges as one-edges. Intuitively, the superscripts 0 and 1
describe whether the edges are used in trees for which the
corresponding bit is 0 or 1, respectively. Note that by this convention zero-cost edges are referred to as one-edges.
Let $last(\bbb^1_i,F)$ and $last(\aaa^1_{i,j},F)$ be
the largest index of an edge from $\bbb^1_i$ and $\aaa^1_{i,j}$,
respectively, that is missing from $F$:
\[
\renewcommand{\arraystretch}{1.3}
\begin{array}{rrcl}
\forall i\in [n]: &\quad last(\bbb^1_i,F) &=& \max (\{0\} \cup \{j \in [rs] \mid b^1_{i,j}
\not\in F \}) \\
\forall i\in [n], \forall j \in [r]: & \quad
last(\aaa^1_{i,j},F) &=& \max (\{0\} \cup \{k \in [s] \mid a^1_{i,j,k}
\not\in F \})
\end{array}
\]

\begin{definition}[$bit_i(F,B)$]
For every $i\in [n]$, every tree $B$, and every $F \subseteq E$ we
say that $bit_i(F,B) = 1$ if:
\begin{itemize}
\item
For all $j \in [rs]$: $b^1_{i,j} \in B$ iff $j > last(\bbb^1_i,F)$.
\item
If $\bbb^1_i \subseteq F$ then for all $j\in [r]$ and $k\in[s]$:
$a^1_{i,j,k}\in B$ iff $k > last(\aaa^1_{i,j},F)$.
\item
If $\bbb^1_i \not\subseteq F$ then for all $j\in [r]$:
$\aaa^1_{i,j}\cap B = \emptyset$.
\end{itemize}
Similarly, we say that $bit_i(F,B) = 0$ if:
\begin{itemize}
\item
$\bbb^1_{i}\cap B = \emptyset$.
\item
For all $j\in [r]$: $\aaa^1_{i,j}\cap B = \emptyset$.
\end{itemize}
\end{definition}

When proving our lower bound, $\RandomFacet$, $\RandomPerm$, and
$\RandomBland$ will be given any initial tree $B_0$ composed
entirely of zero-edges:
\[
B_0 ~\subseteq~ \{a^{0,\ell}_{i,j,k} \mid i\in[n], j\in[r], k \in [s],\ell\in[t]\}
\cup \{b^{0,\ell}_{i,j} \mid i\in[n], j\in[rs],\ell\in[t]\}
\cup \{u^{0,\ell}_i,w^{0,,\ell}_i \mid i \in [n],\ell\in[t]\}
\]
On the other hand, any tree composed entirely of one-edges is
optimal. Hence, for the initial tree all bits are interpreted as
being 0, and for the final, optimal tree all bits are interpreted as
being 1.

\begin{definition}[Functional sets]\label{def:functional}
A subset of edges $F\subseteq E$ is said to be \emph{functional} if
and only if it contains at least one copy of every multi-edge, i.e.,
$F \cap \eee \ne \emptyset$ for all $\eee \in \mathcal{M}$.
\end{definition}

Let $F \subseteq E$ be a subset of edges. It will be convenient to use
the following short-hand notation:
\[
\aaa_i^1 \sqsubseteq F \quad\iff\quad \exists j \in [r]: ~
\aaa_{i,j}^1 \subseteq F
\]
Note that if $r=1$ then $\aaa_i^1 \sqsubseteq F$ has the same meaning
as $\aaa_{i,1}^1 \subseteq F$. When no set $\aaa_{i,j}^1$, for $j\in
[r]$, is completely contained in $F$ we write $\aaa_i^1
\not\sqsubseteq F$.

\begin{definition}[Reset level]
For every $F \subseteq E$ we define the \emph{reset
  level} to be:
\[
reset(F) = \max\left(\{0\} \cup \{i \in [n] \mid \bbb^1_i\subseteq F
\land \aaa^1_i\not\sqsubseteq F\}\right)
\]
\end{definition}

As we shall see, all bits with index lower than~$reset(F)$, for some
functional set $F \subseteq E$, are reset in any optimal tree for
the subgraph $G_F$ defined by $F$.

Recall that an edge $e=(u,v)$ is optimal if and only if it satisfies
$c(u,v) + \val(v) = \val(u)$, where $\val(u)$ is the optimal value of
$u$, and that
a tree is optimal if and only if all its edges are optimal.
For every functional set $F\subseteq E$ we next define a set of
edges $\BB_F \subseteq F$ that we later show is optimal for the subgraph
$G_F$, such that every tree consisting only of edges
from $\BB_F$ is optimal.
The set of edges $\BB_F$ is depicted in figures \ref{case1},
\ref{case2}, \ref{case3}, and \ref{case4}. For simplicity, the
figures show the case when $r=s=t=1$. Edges not in $\BB_F$ are shown as
dotted arrows and edges in $\BB_F$ are shown as unbroken arrows. Edges
that are not in $F$ have been removed; for the general case this should be interpreted as at least one edge missing from the corresponding path.

\begin{figure}[t]
\center
\parbox{3in}{
\center
\includegraphics[scale=0.9]{shortest_paths.3}
\caption{Case $(i)$: $i > reset(F)$ and $\bbb^1_i \subseteq F$}
\label{case1}
}
\hspace*{10pt}
\parbox{3in}{
\center
\includegraphics[scale=0.9]{shortest_paths.4}
\caption{Case $(ii)$: $i > reset(F)$ and $\bbb^1_i \not\subseteq F$}
\label{case2}
}

\vspace*{10pt}

\parbox{3in}{
\center
\includegraphics[scale=0.9]{shortest_paths.5}
\caption{Case $(iii)$: $i = reset(F)$}
\label{case3}
}
\hspace*{10pt}
\parbox{3in}{
\center
\includegraphics[scale=0.9]{shortest_paths.6}
\caption{Case $(iv)$: $i < reset(F)$}
\label{case4}
}

\end{figure}

\begin{definition}[$\BB_F$]\label{D-Sigma-F}
Let $F\subseteq E$ be functional. Define $\BB_F \subseteq F$ to
contain exactly the following edges:
\begin{itemize}
\item[\it(i)]
For all $i > reset(F)$ where $\bbb^1_i \subseteq F$:
\begin{itemize}
\item
For all $j\in [rs]$: ~$b^1_{i,j} \in \BB_F$.
\item
For all $j\in [r]$ and $k\in[s]$:~ $a^1_{i,j,k}\in \BB_F$ if $k >
last(\aaa^1_{i,j},F)$, and $\aaa^0_{i,j,k}\cap F\subseteq \BB_F$ otherwise.
\item
$\uuu^1_i\cap F \subseteq \BB_F$.
\item
$\www^j_i\cap F \subseteq \BB_F$ for all $j$ with $\aaa^1_{i,j} \subseteq F$.
\end{itemize}
\item[\it(ii)]
For all $i > reset(F)$ where $\bbb^1_i \not\subseteq F$:
\begin{itemize}
\item
For all $j \in [rs]$:~ $b^1_{i,j} \in \BB_F$ if $j >
last(\bbb^1_i,F)$, and $\bbb^0_{i,j} \cap F \subseteq \BB_F$ otherwise.
\item
For all $j\in [r]$ and $k\in[s]$:~ $\aaa^0_{i,j,k} \cap F \subseteq \BB_F$.
\item
$\uuu^0_i \cap F \subseteq \BB_F$.
\item
$\www^0_i \cap F \subseteq \BB_F$.
\end{itemize}
\item[\it(iii)]
For $i = reset(F)$:
\begin{itemize}
\item
For all $j\in [rs]$: ~$b^1_{i,j} \in \BB_F$.
\item
For all $j\in [r]$ and $k\in[s]$:~ $a^1_{i,j,k}\in \BB_F$ if $k >
last(\aaa^1_{i,j},F)$, and $\aaa^0_{i,j,k} \cap F \subseteq \BB_F$ otherwise.
\item
$\uuu^1_i\cap F \subseteq \BB_F$.
\item
$\www^0_i \cap F \subseteq \BB_F$.
\end{itemize}
\item[\it(iv)]
For all $i < reset(F)$:
\begin{itemize}
\item
For all $j \in [rs]$:~ $\bbb^0_{i,j}\cap F \subseteq \BB_F$.
\item
For all $j\in [r]$ and $k\in[s]$:~ $\aaa^0_{i,j,k}\cap F \subseteq \BB_F$.
\item
$\uuu^0_i\cap F \subseteq \BB_F$.
\item
$\www^j_i\cap F \subseteq \BB_F$ for all $j \in [r]$.
\end{itemize}
\end{itemize}
\end{definition}

Note that when $i > reset(F)$ and $\bbb^1_i \subseteq F$ then we must
also have $\aaa^1_i \sqsubseteq F$. Hence, in case $(i)$ there exists
a $j\in [r]$ such that $\aaa^1_{i,j} \subseteq F$, which means that
at least one edge in $\BB_F$ is available from $w_i$. It is then not
difficult to check that for every vertex $v$ there is at least one edge in
$\BB_F$ that leaves $v$.
Also note that for every tree $B \subseteq \BB_F$ we have
$bit_i(F,B) = 1$ for all $i \ge reset(F)$, and 
$bit_i(F,B) = 0$ for all $i < reset(F)$.

\begin{lemma}\label{L-Sigma-optimal}
For every functional set $F \subseteq E$, the set of optimal edges
for the subgraph $G_F$ is $\BB_F$.
\end{lemma}

\begin{proof}
In the following we let $\val(u)$ be
the optimal value of $u$ in the graph $G_F$.
Since $G_{n,r,s,t}$ is acyclic then so is $G_F$. The idea of the proof
is to find the distances to the terminal vertex $\TT$ by
backward induction from $\TT$.
In the process of doing so we show that
the edges allowed by $\BB_F$ are exactly the optimal choices. Our
induction hypothesis is that $\val(w_i) = \val(u_i)$ for all
$i > reset(F)$, and that $\val(w_i) = \val(u_i) + 2^{2i}$ for
all $i \le reset(F)$.

We start by making the following two useful observations. We let
$\BB_F^*$ be the actual set of optimal edges, i.e., we hope to show
that $\BB_F = \BB_F^*$.
\begin{claim}\label{claim:1}
Let $i\in[n]$ be given.
If $\val(w_{i+1}) > (2^{2i+1} + 1 + (rs-1)\epsilon) + \val(u_{i+1})$,
then $\bbb^0_{i,j}\cap F \subseteq \BB_F^*$ for
all $j \in [rs]$. Furthermore,
$\val(b_{i,j}) = (2^{2i+1} + 1 + (j-1)\epsilon) +\val(u_{i+1})$ for all 
$j \in [rs]$. In particular, $\val(b_{i,1}) = (2^{2i+1} + 1)
+\val(u_{i+1})$.
\end{claim}
\begin{claim}\label{claim:2}
Let $i\in[n]$ be given.
If $\val(w_{i+1}) < (2^{2i+1} + 1) +\val(u_{i+1})$, then 
$b^1_{i,j} \in \BB_F^*$ for
$j > last(\bbb^1_i,F)$, and $\bbb^0_{i,j} \cap F \subseteq \BB_F^*$ for $j \le
last(\bbb^1_i,F)$. Furthermore,
$\val(b_{i,j}) = \val(w_{i+1})$ for all $j > last(\bbb^1_i,F)$, and 
$\val(b_{i,j}) = (2^{2i+1} + 1 + (j-1)\epsilon) +\val(u_{i+1})$ for all 
$j \le last(\bbb^1_i,F)$. In particular, if $\bbb^1_i \subseteq F$
then $\val(b_{i,1}) = \val(w_{i+1})$, and otherwise $\val(b_{i,1}) =
(2^{2i+1} + 1) +\val(u_{i+1})$.
\end{claim}
To prove the claims observe that the cost of going to $u_{i+1}$ goes
up by $\epsilon$ for every step made along the chain of $b_{i,j}$
vertices. Hence, if it is optimal to leave the chain at some vertex
$b_{i,j}$, then the same is true for all vertices $b_{i,j'}$ with $j'
< j$. In particular, if $\val(w_{i+1}) > (2^{2i} + 1 + (rs-1)\epsilon)
+ \val(u_{i+1})$ then $b_{i,rs}$, and all other $b_{i,j}$ vertices,
should go to $u_{i+1}$. This proves Claim \ref{claim:1}. The same
argument proves Claim \ref{claim:2} for the case where $j \le
last(\bbb^1_i,F)$. When $\val(w_{i+1}) < (2^{2i} + 1) +\val(u_{i+1})$
every $b_{i,j}$ vertex that can reach $w_{i+1}$ would do so with cost
0. This is cheaper than going to $u_{i+1}$, which proves Claim
\ref{claim:2} for the case where $j > last(\bbb^1_i,F)$.
The following two claims are proved in the same way for the
$a_{i,j,k}$ vertices.

\begin{claim}\label{claim:3}
Let $i\in[n]$ and $j \in [r]$ be given.
If $\val(b_{i,1}) > (2^{2i+1} + (s-1)\epsilon) + \val(u_{i+1})$,
then $\aaa^0_{i,j,k} \cap F \subseteq \BB_F^*$ for 
all $k \in [s]$.
Furthermore, $\val(a_{i,j,k}) = (2^{2i+1} + (k-1)\epsilon)
+\val(u_{i+1})$ for all $k \in [s]$. In particular, $\val(a_{i,j,1}) =
2^{2i+1}+\val(u_{i+1})$.
\end{claim}
\begin{claim}\label{claim:4}
Let $i\in[n]$ and $j \in [r]$ be given.
If $\val(b_{i,1}) < 2^{2i+1} +\val(u_{i+1})$, then 
$a^1_{i,j,k} \in \BB_F^*$ for $k > last(\aaa^1_{i,j},F)$, and 
$\aaa^0_{i,j,k} \cap F \subseteq \BB_F^*$ for $k \le last(\aaa^1_{i,j},F)$.
Furthermore,
$\val(a_{i,j,k}) =
\val(b_{i,1})$ for all $k > last(\aaa^1_{i,j},F)$, and 
$\val(a_{i,j,k}) = (2^{2i+1} + (k-1)\epsilon) +\val(u_{i+1})$ for all 
$k \le last(\aaa^1_{i,j},F)$. In particular, if $\aaa^1_{i,j} \subseteq F$
then $\val(a_{i,j,1}) = \val(b_{i,1})$, and otherwise $\val(a_{i,j,1})
= 2^{2i+1} +\val(u_{i+1})$.
\end{claim}

We are now ready to prove the lemma by backward induction in $i$. For
$i = n+1$ we have $\val(\TT) = \val(w_{n+1}) = \val(u_{i+1})$. We
consider four cases, corresponding to the four cases in Definition
\ref{D-Sigma-F}. It is not difficult to verify the claims below by
studying figures \ref{case1},
\ref{case2}, \ref{case3}, and \ref{case4}.

\medskip
\noindent
\textbf{Case \textit{(i)}:}
Assume that $i > reset(F)$ and $\bbb^1_i \subseteq F$. Using claims 
\ref{claim:2} and \ref{claim:4}, and the fact that $\aaa^1_i
\sqsubseteq F$, we see that:
\[
\renewcommand{\arraystretch}{1.3}
\begin{array}{rcll}
\val(b_{i,1}) &=& \val(w_{i+1}) & \\
\val(u_i) &=& \val(b_{i,1}) ~=~ \val(w_{i+1}) &\\
\val(a_{i,j,1}) &=& \val(b_{i,1}) ~=~ \val(w_{i+1}) &\quad \text{for
  all $j \in [r]$ with $\aaa^1_{i,j} \subseteq F$}\\
\val(a_{i,j,1}) &=& 2^{2i+1} + \val(u_{i+1}) &\quad \text{for
  all $j \in [r]$ with $\aaa^1_{i,j} \not\subseteq F$}\\
\val(w_i) &=& \val(a_{i,j,1}) ~=~ \val(b_{i,1}) ~=~ \val(w_{i+1})&
\quad\text{where $\aaa^1_{i,j} \subseteq F$}
\end{array}
\]
and that the optimal edges are those specified by $\BB_F$. Note that
$\val(w_i) = \val(u_i)$, which proves the induction hypothesis.

\medskip
\noindent
\textbf{Case \textit{(ii)}:}
Assume that $i > reset(F)$ and $\bbb^1_i \not\subseteq F$. Using claims 
\ref{claim:2} and \ref{claim:3} we see that:
\[
\renewcommand{\arraystretch}{1.3}
\begin{array}{rcll}
\val(b_{i,1}) &=& (2^{2i+1}+1) + \val(u_{i+1}) & \\
\val(u_i) &=& 2^{2i}+ \val(u_{i+1}) &\\
\val(a_{i,j,1}) &=& 2^{2i+1} + \val(u_{i+1}) &\quad \text{for
  all $j \in [r]$}\\
\val(w_i) &=& 2^{2i}+\val(w_{i+1}) 
\end{array}
\]
and that the optimal edges are those specified by $\BB_F$. Note that
$\val(w_i) = \val(u_i)$, which proves the induction hypothesis.

\medskip
\noindent
\textbf{Case \textit{(iii)}:}
Assume that $i = reset(F)$ such that $\bbb^1_i \subseteq F$ and
$\aaa^1_i \not\sqsubseteq F$. Using claims
\ref{claim:2} and \ref{claim:4} we see that:
\[
\renewcommand{\arraystretch}{1.3}
\begin{array}{rcll}
\val(b_{i,1}) &=& \val(w_{i+1}) & \\
\val(u_i) &=& \val(b_{i,1}) ~=~ \val(w_{i+1}) &\\
\val(a_{i,j,1}) &=& 2^{2i+1} + \val(u_{i+1}) &\quad \text{for
  all $j \in [r]$}\\
\val(w_i) &=& 2^{2i}+\val(w_{i+1})
\end{array}
\]
and that the optimal edges are those specified by $\BB_F$. Note that
$\val(w_i) = 2^{2i}+\val(u_i)$ as desired.

\medskip
\noindent
\textbf{Case \textit{(iv)}:}
Assume that $i < reset(F)$. Note that by induction we have 
$\val(w_{i+1}) = 2^{2(i+1)}+\val(u_{i+1})$.
Using claims 
\ref{claim:1} and \ref{claim:3} we see that:
\[
\renewcommand{\arraystretch}{1.3}
\begin{array}{rcll}
\val(b_{i,1}) &=& (2^{2i+1}+1) + \val(u_{i+1}) & \\
\val(u_i) &=& 2^{2i}+ \val(u_{i+1}) &\\
\val(a_{i,j,1}) &=& 2^{2i+1} + \val(u_{i+1}) &\quad \text{for
  all $j \in [r]$}\\
\val(w_i) &=& \val(a_{i,j,1}) ~=~ 2^{2i+1} + \val(u_{i+1})&\quad \text{for
  any $j \in [r]$}
\end{array}
\]
and that the optimal edges are those specified by $\BB_F$. Note that
$\val(w_i) = 2^{2i}+\val(u_i)$, which proves the induction hypothesis.

This completes the proof.
\end{proof}

\begin{lemma}\label{L-make_switch}
Let $F\subseteq E$ be functional. Then
\begin{itemize}
\item[$(i)$]
Let $i \in [n]$ and $j \in [rs]$ be given. If $b^1_{i,j} \not\in F$,
$b^1_{i,j'} \in F$ for all $j' > j$, $i \ge reset(F)$, and $B
\subseteq \BB_F$, then $b^1_{i,j}$ is an improving switch with respect
to $B$.
\item[$(ii)$]
Let $i \in [n]$, $j \in [r]$, and $k \in [s]$ be given.
If $a^1_{i,j,k} \not\in F$, $a^1_{i,j,k'} \in F$ for all $k' > k$, $i
\ge reset(F)$, $\bbb^1_i \subseteq F$, and $B \subseteq \BB_F$, then
$a^1_{i,j,k}$ is an improving switch with respect to $B$.
\end{itemize}
\end{lemma}

\begin{proof}
We first prove $(i)$. Let $F' = F \cup \{b^1_{i,j}\}$. It is not
difficult to see that $i \ge reset(F')$. Observe also that
$last(\bbb^1_i,F') < j$.
By Definition~\ref{D-Sigma-F}
and Lemma~\ref{L-Sigma-optimal} every tree of $\BB_{F'}$ must
contain $b^1_{i,j}$. Thus, $B \not\subseteq \BB_{F'}$, and there
must be an edge in $F'$
which is an improving switch with respect to $B$. Since no edge
in $F$ is an improving switch, $b^1_{i,j}$ must be an improving
switch with respect to $B$.

The proof of $(ii)$ is similar. Let $F' = F \cup
\{a^1_{i,j,k}\}$. Since $\bbb^1_i \subseteq F$ we again have $i \ge
reset(F')$, and $B \not\subseteq \BB_{F'}$, and it again follows
that $a^1_{i,j,k}$ must be an improving switch with respect to $B$.
\end{proof}

%% file: Lower-RandomFacet.tex
\section{Lower bound for \RandomFacet}\label{S-lower-bound}

Before considering the behavior of the $\RandomFacet$ algorithm when
applied to the lower bound construction $G_{n,r,s,t}$, we first make
some general observations about the algorithm. For this purpose let
$G=(V,E,c)$ be any directed weighted graph, and let $B_0$ be some
initial tree.

The operation of the \RandomFacet\  algorithm may be described by a
binary \emph{computation tree} $T$ as follows.
Each node $u$ of the computation tree corresponds to a
recursive call $\RandomFacet(F(u),B(u))$, where $F(u)\subseteq
E$ is a subset of edges assigned to $u$, and $B(u)$ is a tree
assigned to $u$.
For the root, $u_0$, we have $F(u_0) =
E$ and $B(u_0) = B_0$, where $B_0$ is the initial tree.
For every node $u$ with $F(u)\ne B(u)$,
we assign an edge $e(u) \in F(u) \setminus B(u)$ to $u$. When $F(u)
= B(u)$ we write $e(u) = \bot$. The edge $e(u)$ corresponds to the
edge picked by the $\RandomFacet$ algorithm at the highest level of
the recursion for the recursive call $\RandomFacet(F(u),B(u))$.

Every node~$u$ may have a left child $u_L$ and a
right child $u_R$, corresponding to the first and second recursive
call of the \RandomFacet\ algorithm, respectively.
The left child $u_L$ exists if and only if
$F(u) \ne B(u)$, in which case $F(u_L)=F(u)\setminus
\{e(u)\}$ and $B(u_L)=B(u)$. 
Let $u^*$ be the \emph{rightmost} leaf in the subtree of~$u$. It is
obtained by following a path from~$u$ that makes a right turn whenever
possible until reaching a leaf. It can also be defined recursively as
follows. If $u$ is a leaf then $u^*=u$. If $u_R$ exists, then
$u^*=(u_R)^*$. Otherwise, $u^*=(u_L)^*$. Note that $B(u^*)$ is
the tree returned by the recursive call of \RandomFacet\ at
$u$. The correctness of the algorithm thus implies that
$B(u^*)$ is an optimal shortest path tree for the graph defined by $F(u)$.
The right child~$u_R$ exists if and only if~$u_L$ exists and $e(u)$ is
an improving switch with respect to $B((u_L)^*)$. 
If $u_R$ exists,
then $F(u_R)= F(u)$ and $B(u_R)=B((u_L)^*)[e(u)]$. 

The \RandomFacet\ algorithm is of course a randomized algorithm, so
the computation tree it defines is not unique. Instead, the
computation tree is a random variable, and
$\RandomFacet$ defines a probability distribution over
computation trees. The random choices made
by the algorithm manifest themselves in the edges $u(e)$ assigned to
the nodes of the computation tree.
Note that every right-edge $(u,u_R)$ in the computation tree
corresponds to an improving switch performed by the $\RandomFacet$
algorithm. Since we are interested in the expected number of improving
switches performed this motivates the following definition. For every
computation tree $T$, let $switch(T)$ be the set of nodes $u$ for which
there exists a right child $u_R$. Then the expected number of
improving switches performed by the $\RandomFacet$ algorithm is
$\E{|switch(T)|}$.

Let us note that every element $u$ of $switch(T)$ can be uniquely
identified with the path from the root to $u$. Such a path can be
described by a sequence of $L$ and $R$ labels. Since $u$ has a right
child, the path obtained by appending an additional $R$ label must
also appear in $T$. Furthermore, the sets of edges assigned to nodes
along the path are uniquely determined by the labels and the
encountered edges picked by the $\RandomFacet$ algorithm. It will
therefore be helpful to include the picked edges in the description of
the path. This leads us to the following definition.

\begin{definition}[\bf Computation path]
A \emph{computation path} $P$ is a sequence of pairs $\langle
(e_0,d_0),(e_1,d_1),\ldots,(e_k,d_k) \rangle$, where
$e_0,e_1,\ldots,e_k\in E$ are distinct edges of the graph $G$, and
where $d_0,d_1,\ldots,d_{k}\in \{L,R\}$. We let $F_0,F_1,\dots,F_{k+1}$ be
the sets of edges assigned to the nodes of the path, such that $F_0 =
E$, and for all $\ell \le k$ we have $F_{\ell+1} = 
F_\ell\setminus\{e_\ell\}$ if $d_\ell = L$ and $F_{\ell+1}=F_\ell$ otherwise.
Furthermore, we let $paths(G)$ be the set
of all computation paths for $G$ for which $d_k=R$, i.e., we require the
paths of $paths(G)$ to end with a right-edge. 
\end{definition}

If a computation path $P$ appears in a computation
tree $T$ we write $P \in T$. 
Furthermore, we let $\mathcal{T}_{E,B_0}$ be the set
of all possible computation trees generated by a call to
$\RandomFacet(E,B_0)$, and $p_{E,B_0}(T)$ be the probability that
the computation tree $T$ is generated. Using that every right-edge of a
computation tree corresponds uniquely to a path, we get the following
useful lemma. This technique was first used by
G{\"a}rtner~\cite{Gartner02}.

\begin{lemma}\label{lemma:count_paths}
$\E{|switch(T)|} \;=\; \sum_{P\in paths(G)} \PPr{P \in T}$.
\end{lemma}
\begin{proof}
Let $\id_{P \in T}$ be an indicator variable that is 1 if $P \in T$
and 0 otherwise. Then we have:
\begin{align*}
\E{|switch(T)|} &~=~ \sum_{T \in \mathcal{T}_{E,B_0}} p_{E,B_0}(T) \,|switch(T)|\\
&~=~
\sum_{T \in \mathcal{T}_{E,B_0}} p_{E,B_0}(T) \,|\{P \in
paths(G) \mid P \in T\}| \\
&~=~
\sum_{P \in paths(G)}
\sum_{T \in \mathcal{T}_{E,B_0}} p_{E,B_0}(T)\, \id_{P \in T}\\
&~=~
\sum_{P \in paths(G)}
\PPr{P \in T}\;.
\end{align*}
\end{proof}

We next focus on the behavior of the $\RandomFacet$ algorithm when
applied to the graph $G = G_{n,r,s,t}$ starting from a tree $B_0$
consisting entirely of zero-edges. $n \in \NN$ is the number of bits
of the corresponding binary counter. $r,s,t \in \NN$ are parameters that
will be specified by the analysis to ensure that the expected number
of improving switches performed by the
$\RandomFacet$ algorithm is at least ${\rm e}^{\Omega(\sqrt{n})}$. We
show that it suffices to use $r = \Theta(\log n)$, $s =
\Theta(\sqrt{n}\log n)$, and $t = \Theta(\log n)$.
Since the graph $G_{n,r,s,t}$ contains
$N=\Theta(nrs)=\Theta(n^{3/2}\log^2 n)$ vertices and
$M=\Theta(nrst)=\Theta(n^{3/2}\log^3 n)$ edges, we prove the following
theorem.

\begin{theorem}\label{thm:main}
The expected number of switches performed by the
\RandomFacet\ algorithm, when run on the graph $G_{n,r,s,t}$, where 
$r = \Theta(\log n)$, $s = \Theta(\sqrt{n}\log n)$, and $t =
\Theta(\log n)$, with initial tree $B_0$, is at least ${\rm
  e}^{\Omega(\sqrt{n})} = {\rm e}^{\Omega(M^{1/3}/\log M)}$,
where $M$ is the number of edges in $G_{n,r,s,t}$.
\end{theorem}

First, we need some definitions that allow us to identify interesting events
along a computation path $P=\langle
(e_0,d_0),(e_1,d_1),\ldots,(e_k,d_k) \rangle$. Note that the edges
$e_0,e_1,\ldots,e_k$ are all distinct. For every $e\in E$, we define
$\sigma_P(e)$ to be the index of $e$ in $P$, if it appears in~$P$, 
and $\infty$ otherwise:
\[
\sigma_P(e) ~=~ \begin{cases}
\ell & \text{if $e = e_\ell$ for some $\ell \in \{0,\dots,k\}$}\\
\infty & \text{otherwise}
\end{cases}
\]
It will again be helpful to work with sets of edges rather than single
edges. We therefore define $\sigma_P(\bbb^1_i)$ to be the first
occurrence of an edge from $\bbb^1_i$ in $P$. We also define 
$\sigma_P(\aaa^1_i)$ to be the smallest index for which at least one
edge from every set $\aaa^1_{i,j}$ has appeared in $P$. Note that
removing an edge from $\bbb^1_i$ makes it impossible to reach
$w_{i+1}$ from $b_{i,1}$, and, similarly, removing at least one edge
from every set $\aaa^1_{i,j}$, for all $j\in[r]$, makes it impossible
to reach $b_{i,1}$ from any vertex $a_{i,j,1}$, for $j\in
[r]$. Formally, we define:
\[
\renewcommand{\arraystretch}{1.5}
\begin{array}{rrcl}
\forall i \in [n]: &
\sigma_P(\bbb^1_i) &\!=\!& \displaystyle\min_{j \in [rs]} ~ \sigma_P(b^1_{i,j}) \\
\forall i \in [n], \forall j \in [r]: &
\sigma_P(\aaa^1_{i,j}) &\!=\!& \displaystyle\min_{k\in[s]} ~ \sigma_P(a_{i,j,k}^1)\\
\forall i \in [n]: &
\sigma_P(\aaa^1_i) &\!=\!& \displaystyle\max_{j \in [r]} ~
\displaystyle\min_{k\in[s]} ~ \sigma_P(a_{i,j,k}^1)\\
\end{array}
\]

Next, we define computation paths for which the random selection of edges is
well-behaved. 

\begin{definition}[\bf Canonical paths] \label{D-canonical}
Let $n\ge i_1>i_2>\cdots>i_p\ge 1$, let $S = \{i_1,i_2,\ldots,i_p\}
\subseteq [n]$, let
$P=\langle(e_0,d_0),(e_1,d_1),\ldots,(e_k,d_k) \rangle$ be a
computation path, and let $F_0,F_1,\dots,F_{k+1}$ be the sets of edges
assigned to the nodes of the path.
We say that $P$ is \emph{$(i_1,i_2,\ldots,i_p)$-canonical} if and only
if it satisfies:
\begin{itemize}
\item[$(i)$]
$\sigma_P(\bbb^1_{i_q}) < \sigma_P(\bbb^1_{i_{q+1}})$ for all
  $q \in [p-1]$.
\item[$(ii)$]
$\sigma_P(\bbb^1_i) \le \sigma_P(\aaa^1_i)$ for all $i \in [n]$.
\item[$(iii)$]
$F_\ell$ is functional for all $\ell \in \{0,\dots,k+1\}$.
\item[$(iv)$]
For every $\ell \in \{0,\dots,k\}$, $d_\ell = R$ if and only if $\ell
\in \mathcal{R}_{P,S}$, where:
\[
\mathcal{R}_{P,S} ~:=~ \{\sigma_P(e) \mid e \in \bbb^1_{i_q}, q \in [p]\} \cup
\{\sigma_P(e) \mid e \in \aaa^1_{i_q,j}, q\in [p], j\in [r],
\aaa^1_{i_q} \not\sqsubseteq F_{\sigma_P(e)} \setminus \{e\}\}\;.
\]
Also, $\sigma_P(\aaa^1_{i_p}) = k$ such that $d_k = R$.
\end{itemize}
We let $canon_S(G)$ be the set of $(i_1,i_2,\ldots,i_p)$-canonical
computation paths.
\end{definition}

Note that in Definition \ref{D-canonical} $(ii)$ weak inequality is used to allow that $\sigma_P(\bbb^1_i) = \sigma_P(\aaa^1_i) = \infty$ for some $i$.
Since every canonical path ends with a right-edge we have
$canon_S(G) \subseteq paths(G)$ for every $S \subseteq [n]$.
Also note that $canon_S(G) \cap canon_{S'}(G) = \emptyset$ for $S \ne
S'$. Recall that $P \in T$ is the event that the computation path $P$
appears in a random computation tree $T$ generated by the
$\RandomFacet$ algorithm. Let $p = \lfloor\sqrt{n}\rfloor$.
It follows from Lemma
\ref{lemma:count_paths} that:
\begin{align*}
\E{|switch(T)|} &~=~ \sum_{P\in paths(G)} \PPr{P \in T} \\
&~\ge~
\sum_{S \subseteq [n]} ~\sum_{P \in canon_S(G)} \PPr{P \in T} \\
&~\ge~
\sum_{S \subseteq [n]: |S| = p} ~\sum_{P \in canon_S(G)} \PPr{P
  \in T} \;.
\end{align*}
We use this inequality to prove Theorem \ref{thm:main}, i.e., we show
that:
\[
\sum_{S \subseteq [n]: |S| = p} ~\sum_{P \in canon_S(G)} \PPr{P
  \in T} ~=~ {\rm e}^{\Omega(\sqrt{n})}\;.
\]

The main technical lemma of this section, from which the desired lower bound on the expected number of steps performed by the \RandomFacet\ algorithm easily follows, is the following lemma.

\begin{lemma}\label{L-appear-S}
For every $S\subseteq[n]$,  where $|S| = p = \lfloor\sqrt{n}\rfloor$, we have
\[
\sum_{P \in canon_S(G)} \PPr{P \in T} ~\ge~
\frac{1}{2\, p!} \;.
\]
\end{lemma}

The proof of Lemma~\ref{L-appear-S} is quite involved and is the main technical contribution of this section. Before presenting this proof, we show that Lemma~\ref{L-appear-S} allows us to obtain the subexponential lower bound we are after.

\medskip\noindent
\textbf{Proof of Theorem \ref{thm:main}:\ }
Let $p = \lfloor\sqrt{n}\rfloor$. We get from Lemma
\ref{lemma:count_paths} and Lemma \ref{L-appear-S} that:
\begin{align*}
\E{|switch(T)|} ~\ge~ 
\sum_{S \subseteq [n]: |S|=p} ~\sum_{P \in canon_S(G)} \PPr{P
  \in T}
~\ge~
\sum_{S \subseteq [n]: |S|=p} \frac{1}{2\, |S|!}
~=~
\frac{1}{2}\frac{1}{p!} { n \choose p }\;.
\end{align*}
It follows that $\E{|switch(T)|} \;\ge\; {\rm e}^{\Omega(\sqrt{n})}$ because $\frac{1}{2}\frac{1}{p!} { n \choose p } \;\ge\; {\rm e}^{(2-o(1))\sqrt{n}}$, as

\[
\frac{1}{p!} {n \choose p} 
~=~ \frac{n!}{(p!)^2(n-p)!} 
~\ge~ \frac{(n-p)^{p}}{(p!)^2} 
~=~ \frac{\Bigl(n\Bigl(1-\frac{1}{p}\Bigr)\Bigr)^{p}}{(p!)^2}
~\sim~ \frac{\frac{n^{\sqrt{n}}}{\rm e}}{2\pi \sqrt{n} \Bigl(\frac{n}{{\rm e}^2} \Bigr)^{\sqrt{n}}}
~=~ \frac{{\rm e}^{2\sqrt{n}}}{2\pi {\rm e} \sqrt{n}}\;.
\]
\qed

The remainder of this section is devoted to the proof of
Lemma~\ref{L-appear-S}. Throughout the remainder of the section we let
$S = \{i_1,i_2,\ldots,i_p\} \subseteq [n]$, where
$i_1>i_2>\cdots>i_p$, be fixed.

Let $T$ be a computation tree. Observe that in order to check whether
$T$ contains an 
$(i_1,i_2,\ldots,i_p)$-canonical computation path it suffices to
follow a single path in $T$ starting from the root. Indeed, Definition
\ref{D-canonical} $(iv)$ specifies at every step whether to go to the
left child or the right child. More precisely, let $P
=\langle(e_0,d_0),(e_1,d_1),\ldots,(e_k,d_k) \rangle \in T$ be a
computation path that is a proper prefix of some
$(i_1,i_2,\ldots,i_p)$-canonical path. Let $u(P)$ be the last node
on $P$, and let $P' =
\langle(e_0,d_0),(e_1,d_1),\ldots,(e_k,d_k),(e(u(P)),L) \rangle$ be the
path obtained by extending $P$ with the edge picked at $u(P)$. The
choice to use $L$ here is arbitrary. Then $P$ can only be extended into an 
$(i_1,i_2,\ldots,i_p)$-canonical path by going to the right child of
$u$ if $k+1 \in \mathcal{R}_{S,P'}$ and to the left child of $u$
otherwise. We define:
\[
d(P) ~=~ \begin{cases}
R & \text{if $k+1 \in \mathcal{R}_{S,P'}$}\\
L & \text{otherwise}\\
\end{cases}
\]
Hence, we can check whether $T$ contains an
$(i_1,i_2,\ldots,i_p)$-canonical path as follows. Starting with the
empty path $P$, we repeatedly append the pair $(e(u(P)),d(P))$ to
$P$. We stop when either $P$ is $(i_1,i_2,\ldots,i_p)$-canonical;
Definition \ref{D-canonical} $(i)$, $(ii)$, or $(iii)$ show that $P$
can not be extended to an $(i_1,i_2,\ldots,i_p)$-canonical path; or
there is no child of $u(P)$ in direction $d(P)$. In fact, Lemma
\ref{lemma:policies} below shows that the last case never happens.
This procedure is
guaranteed to find an $(i_1,i_2,\ldots,i_p)$-canonical path in $T$ if
it exists. Note also that $T$ can contain at most one
$(i_1,i_2,\ldots,i_p)$-canonical path.

\begin{definition}[$P_S(T)$]
Let $T$ be a computation tree, and let $P$ be the maximal path in $T$ that
is a proper prefix of an $(i_1,i_2,\ldots,i_p)$-canonical path. We
define $P_S(T)$ to be the path obtained by appending the pair
$(e(u(P)),d(P))$ to $P$.
\end{definition}

Note that if $T$
contains an $(i_1,i_2,\ldots,i_p)$-canonical path then $P_S(T)$ is exactly
this path. Also note that when $T$ is generated at random then
$P_S(T)$ is a random variable. We next define events
corresponding to the different ways in which we can fail to find an
$(i_1,i_2,\ldots,i_p)$-canonical path in $T$. These events correspond
directly to the requirements in Definition \ref{D-canonical}. 

\begin{definition}[\bf Bad events]\label{def:bad}
Let $T$ be generated at random by the $\RandomFacet$ algorithm.
Suppose that $P_S(T) = \langle (e_0,d_0),(e_1,d_1),\ldots,(e_k,d_k)
\rangle$, and that $F_0,F_1,\dots,F_{k+1}$ are the sets of edges
assigned to the nodes of the path.
\begin{itemize}
\item[$(i)$]
For every $q \in [p-1]$, define $Bad_1(q)$ to be the event that
$\sigma_{P_S(T)}(\bbb^1_{i_q}) = \infty$ and
$\sigma_{P_S(T)}(\bbb^1_{i_{q'}}) = k$ for some $q' > q$. Also
define $Bad_1 := \bigcup_{q\in [p-1]} Bad_1(q)$.
\item[$(ii)$]
For every $i \in [n]$, define $Bad_2(i)$ to be the event that 
$\sigma_{P_S(T)}(\bbb^1_{i}) = \infty$ and
$\sigma_{P_S(T)}(\aaa^1_{i}) = k$.
Also define $Bad_2 := \bigcup_{i\in [n]} Bad_2(i)$.
\item[$(iii)$]
For every $\eee \in \mathcal{M}$, define $Bad_3(\eee)$ to be the event that $F_{k+1}$
is not functional because $F_{k+1} \cap \eee = \emptyset$.
Also define $Bad_3 := \bigcup_{\eee\in \mathcal{M}} Bad_3(\eee)$.
\end{itemize}
We let $Good_1$, $Good_2$, and $Good_3$ be the
complements of $Bad_1$, $Bad_2$, and $Bad_3$ respectively.
\end{definition}

Note that the events $Bad_2$ and $Bad_3$ are disjoint since $Bad_2$
only occurs when $e_k \in \aaa^1_{i,j}$ for some $i$ and $j$, and
$Bad_3$ only occurs when $e_k \in \eee$ for some $\eee \in
\mathcal{M}$.

We now make the first
step towards proving Lemma \ref{L-appear-S}. Recall that $\sum_{P \in
  canon_S(G)} \PPr{P \in T}$ is the probability that the computation
tree $T$ generated by $\RandomFacet(E,B_0)$ contains an
$(i_1,i_2,\ldots,i_p)$-canonical path. Following the above discussion
we know that this happens if and only if none of the events $Bad_1$,
$Bad_2$, or $Bad_3$ occur. Since $Bad_2$ and $Bad_3$ are
disjoint we get:
\begin{align*}
\sum_{P \in
  canon_S(G)} \PPr{P \in T} &~=~ 
\PPr{ Good_1 \land Good_2 \land Good_3 } \\
&~=~ \PPr{ Good_1 }\, \PPr{ Good_2 \land Good_3 \mid Good_1 } \\
&~=~
\PPr{ Good_1 }\, ( 1 - \PPr{ Bad_2 \mid Good_1 } - \PPr{ Bad_3 \mid
  Good_1})
\end{align*}

Before estimating $\PPr{ Good_1 }$, $\PPr{ Bad_2 \mid Good_1 }$, and 
$\PPr{ Bad_3 \mid Good_1}$, we start by proving a couple of lemmas that
describe the trees assigned to nodes along the path 
$P_S(T)$, and show that we always have $P_S(T) \in T$. Recall that $S
= \{i_1,i_2,\dots,i_p\}$.

\begin{lemma}\label{lemma:F}
Let $T$ be a computation tree, let 
\begin{align*}
P &~=~ P_S(T) ~=~ \langle
(e_0,d_0),(e_1,d_1),\ldots,(e_k,d_k) \rangle\;, \\
P' &~=~ \langle (e_0,d_0),(e_1,d_1),\ldots,(e_{k-1},d_{k-1})
\rangle\;,
\end{align*}
and let $F_0,F_1,\dots,F_{k}$ be the sets of edges assigned to nodes
on $P'$. Then:
\begin{itemize}
\item[$(i)$]
For all $q \in [p]$ and $\ell \in \{0,\dots,k\}$: $\bbb^1_{i_q}
\subseteq F_\ell$ and $\aaa^1_{i_q} \sqsubseteq F_\ell$.
\item[$(ii)$]
For all $\ell \in \{0,\dots,k\}$: $reset(F_\ell) = 0$.
\end{itemize}
\end{lemma}

\begin{proof}
$(i)$ follows from Definition \ref{D-canonical} $(iv)$, i.e., every
  time a right-edge is used on $P$ no edge is removed, and
  right-edges are used exactly in a way that ensures $(i)$.

We next prove $(ii)$. Let $\ell \in \{0,\dots,k\}$ be given. For all
$i \not\in S$, Definition \ref{D-canonical} $(ii)$ shows that
$\sigma_P(\bbb^1_i) \le \sigma_P(\aaa^1_i)$. In particular, if
$\aaa^1_i \not\sqsubseteq F_\ell$ then we also have $\bbb^1_i
\not\subseteq F_\ell$. When this is combined with $(i)$ we see that
$reset(F_\ell) = 0$.
\end{proof}

\begin{lemma}\label{lemma:policies}
Let $T$ be a computation tree, let 
\begin{align*}
P &~=~ P_S(T) ~=~ \langle
(e_0,d_0),(e_1,d_1),\ldots,(e_k,d_k) \rangle;, \\
P' &~=~ \langle (e_0,d_0),(e_1,d_1),\ldots,(e_{k-1},d_{k-1})
\rangle\;,
\end{align*}
and let $u_0,u_1,\dots,u_{k}$ be the nodes on the path $P'$. 
Define $i_0 := n+1$, $\sigma_P(\aaa^1_{i_0}) := -1$, and
$\sigma_P(\bbb^1_{i_{p+1}}) := \infty$. Then:
\begin{itemize}
\item[$(i)$]
For all $q \in [p]$:
$\sigma_{P'}(\bbb^1_{i_q}) ~\le~ \sigma_{P'}(\aaa^1_{i_q}) ~\le~
\sigma_{P'}(\bbb^1_{i_{q+1}})$.
\item[$(ii)$]
For all $q \in [p]$ and all
$\sigma_{P'}(\aaa^1_{i_{q-1}}) < \ell \le \sigma_{P'}(\bbb^1_{i_q})$ we
have:
\begin{itemize}
\item
For all $i > i_{q-1}$: $bit_i(F(u_\ell),B(u_\ell)) = 1$. 
\item
$\bbb^1_{i_{q-1}} \subseteq B(u_\ell)$.
\item
For all $i < i_{q-1}$: $bit_i(F(u_\ell),B(u_\ell)) = 0$. 
\end{itemize}
We refer to the interval $\sigma_{P}(\aaa^1_{i_{q-1}}) < \ell \le
\sigma_{P}(\bbb^1_{i_q})$ as the $\bbb_q$-phase.
\item[$(iii)$]
For all $q \in [p]$ and all
$\sigma_{P'}(\bbb^1_{i_{q}}) < \ell \le \sigma_{P'}(\aaa^1_{i_q})$ we have:
\begin{itemize}
\item
For all $i \ne i_q$: $bit_i(F(u_\ell),B(u_\ell)) = 1$. 
\item
$\aaa^1_{i_{q},j} \cap B(u_\ell) = \emptyset$ for all $j \in [r]$.
\end{itemize}
We refer to the interval $\sigma_P(\bbb^1_{i_{q}}) < \ell \le
\sigma_P(\aaa^1_{i_q})$ as the $\aaa_q$-phase.
\item[$(iv)$]
$P \in T$.
\end{itemize}
\end{lemma}

\begin{proof}
We start by partitioning the indices $\ell \in \{0,\dots,k\}$ into
phases. We later show that the phases alternate as described by the
lemma. Let $\mathcal{R} = \{\sigma_P(\aaa^1_{i_{q}}),
\sigma_P(\bbb^1_{i_q}) \mid q\in [p]\}$ be the
set of critical moments when we transition from one phase to
another. We say that $\ell$ is in the $\bbb_q$-phase if the largest
element of $\mathcal{R}$ smaller than $\ell$ is
$\sigma_P(\aaa^1_{i_{q-1}})$, i.e., we later prove that this phase
leads to an edge from $\bbb^1_{i_q}$ being picked. 
If no element of $\mathcal{R}$ is smaller than $\ell$ we say that
$\ell$ is in the $\bbb_1$-phase. Similarly, we say
that $\ell$ is in the $\aaa_q$-phase if the largest
element of $\mathcal{R}$ smaller than $\ell$ is
$\sigma_P(\bbb^1_{i_{q}})$.

We prove $(i)$, $(ii)$, and $(iii)$ jointly by induction in
$\ell$. For $\ell = 0$, $\ell$ is in the $\bbb_1$-phase, $F(u_0) = E$,
and $B(u_0) = B_0$, a tree consisting entirely of zero-edges. Hence,
$(ii)$ is satisfied for $\ell = 0$. Assume that the lemma is satisfied
for some given $\ell < k$. To prove that the lemma also holds for
$\ell+1$ we consider five cases:
\begin{enumerate}
\item
$d_\ell = L$.
\item
$d_\ell = R$, $e_\ell \in \bbb^1_{i_{q'}}$, and
$\ell$ is in the $\bbb_q$-phase, for some $q,q' \in [p]$.
\item
$d_\ell = R$, $e_\ell \in \aaa^1_{i_{q'},j}$, and
$\ell$ is in the $\bbb_q$-phase, for some $q,q' \in [p]$ and $j \in
  [r]$.
\item
$d_\ell = R$, $e_\ell \in \bbb^1_{i_{q'}}$, and
$\ell$ is in the $\aaa_q$-phase, for some $q,q' \in [p]$.
\item
$d_\ell = R$, $e_\ell \in \aaa^1_{i_{q'},j}$, and
$\ell$ is in the $\aaa_q$-phase, for some $q,q' \in [p]$ and $j \in
  [r]$.
\end{enumerate}

\noindent
\textbf{Case 1:}
Assume that $d_\ell = L$. 
Recall from
Definition \ref{D-canonical} that $d_\ell = R$ if and only if $\ell
\in \mathcal{R}_{P,S}$ where $\mathcal{R} \subseteq
\mathcal{R}_{P,S}$. For the case when 
$d_\ell = L$ we therefore have $\ell \not\in \mathcal{R}$, which means
that the phase did not change from $\ell$ to $\ell+1$. We also
have $B(u_{\ell+1}) = B(u_\ell)$. Recall that $F(u_{\ell+1}) =
F(u_{\ell}) \setminus \{e_\ell\}$. Since $e_\ell \not\in B(u_\ell)$
we have $bit_i(F(u_\ell) \setminus \{e_\ell\}, B(u_\ell)) = 1$ if 
$bit_i(F(u_\ell), B(u_\ell)) = 1$, and $bit_i(F(u_\ell) \setminus
\{e_\ell\}, B(u_\ell)) = 0$ if $bit_i(F(u_\ell), B(u_\ell)) =
0$. This completes the induction step.

\medskip\noindent
\textbf{Case 2:}
Let $q,q' \in [p]$, and assume that $d_\ell = R$, that $e_\ell \in
\bbb^1_{i_{q'}}$, and that $\ell$ is in the $\bbb_q$-phase. 
For all $q'' < q-1$ we have $bit_{i_{q''}}(F(u_\ell),B(u_\ell)) 
= 1$, and since Lemma \ref{lemma:F} shows that $\bbb^1_{i_{q''}}
\subseteq F(u_\ell)$, it follows that $e_\ell \not\in
\bbb^1_{i_{q''}}$. Since $\bbb^1_{i_{q-1}} \subseteq B(u_\ell)$ it must
then be the case that $q'\ge q$. Since $P'$ satisfies Definition
\ref{D-canonical} $(i)$ it follows
that $q' = q$. Hence, $\ell = \sigma_P(\bbb^1_{i_q})$, and we
transition to the $\aaa_q$-phase. 

Let $B = B(((u_\ell)_L)^*)$ be
the tree returned by the first recursive call of the $\RandomFacet$
algorithm at $u_\ell$. Since $P'$ satisfies Definition
\ref{D-canonical} $(iii)$, $F(u_\ell) \setminus \{e_\ell\}$ is
functional, and Lemma \ref{L-Sigma-optimal} shows that $B
\subseteq \BB_{F(u_\ell) \setminus \{e_\ell\}}$. From Lemma
\ref{lemma:F} we know that $reset(F(u_\ell) \setminus \{e_\ell\}) =
0$. Lemma \ref{L-make_switch} then
shows that $e_\ell$ is an improving switch w.r.t. $B$. 
Moreover, since $reset(F(u_\ell) \setminus \{e_\ell\}) = 0$ we get
from Definition \ref{D-Sigma-F} that $B[e_\ell]$ has the form
described in $(iii)$. Note in particular that $\bbb^1_{i_q}
\not\subseteq F(u_\ell) \setminus \{e_\ell\}$.

\medskip\noindent
\textbf{Case 3:}
Let $q,q' \in [p]$ and $j \in [r]$, and assume that $d_\ell = R$, that
$e_\ell \in \aaa^1_{i_{q'},j}$, and that $\ell$ is in the
$\bbb_q$-phase. From the definition of $\mathcal{R}_{P,S}$ it follows
that $\aaa^1_{i_{q'}} \not\sqsubseteq F(u_\ell) \setminus \{e_\ell\}$.
We first prove that $q' = q-1$, which implies that $\ell+1$ is also in
the $\bbb_q$-phase. For $q'' < q-1$ we have
$bit_{i_{q''}}(F(u_\ell), B(u_\ell)) = 1$, and Lemma \ref{lemma:F}
shows that $\aaa^1_{i_{q''}} \sqsubseteq F(u_\ell)$. 
Since $P'$ satisfies Definition \ref{D-canonical} $(ii)$ it then
follows that $q' = q-1$. Hence, we have $\aaa^1_{i_{q-1}}
\not\sqsubseteq F(u_\ell) \setminus \{e_\ell\}$ and since
$\bbb^1_{i_{q-1}} \subseteq B(u_\ell)$ we get
$reset(F(u_\ell) \setminus \{e_\ell\}) = i_{q-1}$.

The remainder of the proof is similar to the last half of the proof of
case 2. Let $B = B(((u_\ell)_L)^*)$ be
the tree returned by the first recursive call of the $\RandomFacet$
algorithm at $u_\ell$. Since $P'$ satisfies Definition
\ref{D-canonical} $(iii)$, $F(u_\ell) \setminus \{e_\ell\}$ is
functional, and Lemma \ref{L-Sigma-optimal} shows that $B
\subseteq \BB_{F(u_\ell) \setminus \{e_\ell\}}$. Since
$reset(F(u_\ell) \setminus \{e_\ell\}) = i_{q-1}$,  
Lemma \ref{L-make_switch}
shows that $e_\ell$ is an improving switch w.r.t. $B$. 
Moreover, we get
from Definition \ref{D-Sigma-F} that $B[e_\ell]$ has the form
described in $(ii)$.

\medskip\noindent
\textbf{Case 4:}
Let $q,q' \in [p]$, and assume that $d_\ell = R$, that $e_\ell \in
\bbb^1_{i_{q'}}$, and that $\ell$ is in the $\aaa_q$-phase. 
Since for all $q'' \ne q$ we have
$bit_{i_{q''}}(F(u_\ell),B(u_\ell)) = 1$ and, by Lemma 
\ref{lemma:F}, $\bbb^1_{i_{q''}} \subseteq F(u_\ell)$ such that 
$\bbb^1_{i_{q''}} \subseteq B(u_\ell)$, we must have $q' = q$.
In particular, $\ell+1$ is also in the $\aaa_q$-phase. The remainder
of the proof is the same as the last half of the proof of case 2.

\medskip\noindent
\textbf{Case 5:}
Let $q,q' \in [p]$ and $j \in [r]$, and assume that $d_\ell = R$, that
$e_\ell \in \aaa^1_{i_{q'},j}$, and that $\ell$ is in the
$\aaa_q$-phase.
Since for all $q'' \ne q$ we have
$bit_{i_{q''}}(F(u_\ell),B(u_\ell)) = 1$ and, by Lemma 
\ref{lemma:F}, $\aaa^1_{i_{q''}} \sqsubseteq F(u_\ell)$ such that 
$\aaa^1_{i_{q''}} \sqsubseteq B(u_\ell)$, we must have $q' = q$.
In particular, $\ell= \sigma_P(\aaa^1_{i_q})$ such that $\ell+1$ is in
the $\bbb_{q+1}$-phase. The remainder
of the proof is the same as the last half of the proof of case 3.

\medskip

It remains to prove $(iv)$: $P \in T$. Per definition we have $P' \in
T$, so it suffices to show that $u_k$ has a child in direction $d_k$.
For $d_k = L$ we must show that there is a vertex from which more than
one edge is available in $F(u_\ell)$. This follows, for instance, from
the fact that $F(u_\ell)$ is functional such that any $w_i$ vertex,
for $i\in [n]$, has at least two out-going edges in $F(u_\ell)$. When
$d_k = R$ we proved in cases 2 to 5 that $e_\ell$ is an improving
switch with respect to $B = B(((u_\ell)_L)^*)$, which shows that
$u_k$ has a right child.

\end{proof}

\begin{lemma}\label{lemma: good 3} 
$\Pr[Good_1] \ge \frac{1}{p!}$.
\end{lemma}

\begin{proof}
Let $Good_1(q)$ be the complement of the event $Bad_1(q)$ for every $q
\in [p-1]$. Observe first that:
\[
\PPr{Good_1} ~=~ \prod_{q=1}^{p-1} \PPr{Good_1(q) \mid
  Good_1(1),\dots,Good_1(q-1)}\;.
\]
We will prove that $\Pr[Good_1(q) \mid Good_1(1),\dots,Good_1(q-1)]
\ge \frac{1}{p-q+1}$ for every $q \in [p-1]$, from which it follows
that:
\[
\PPr{Good_1} ~\ge~ \prod_{q=1}^{p-1} \frac{1}{p-q+1} ~=~ \frac{1}{p!}\;.
\]

Let $q \in [p-1]$ be given, let $T$ be a computation tree
generated by $\RandomFacet(E,B_0)$, and let $u_0,\dots,u_{k+1}$ be
the nodes on the path $P_S(T) = \langle
(e_0,d_0),(e_1,d_1),\ldots,(e_k,d_k) \rangle$. Let $\ell$ be the index
of the first occurence of an edge picked from $\bbb^1_{i_{q'}}$ along
$P_S(T)$ for some
$q' \ge q$, i.e., $\ell = \min \{\sigma_{P_S(T)}(\bbb^1_{i_{q'}}) \mid
q' \ge q\}$. If no such edge was picked along 
$P_S(T)$ then we get the event $Good_1(q)$, and we may therefore
assume that $\ell$ exists. The index $\ell$ must be part of a
$\bbb$-phase, since, according to Lemma \ref{lemma:policies}, the edge
$e_\ell$ would otherwise be part of the current policy $B(u_\ell)$.
In fact, since we condition on $Good_1(1),\dots,Good_1(q-1)$, $\ell$ must
be part of the $\bbb_q$-phase.
From Lemma \ref{lemma:policies} we then know that $\bbb^1_{i_{q'}}
\cap B(u_\ell) = \emptyset$ for all $q' \ge q$. Since the sets
$\bbb^1_{i_{q'}}$ all have the same size, the probability that the
edge picked at $u_\ell$ was
from $\bbb^1_q$ is
$\frac{1}{p-q+1}$. In this case we again get the event $Good_1(q)$.
\end{proof}

\begin{lemma}\label{lemma:bad_2} For every $i \in [n]$ we have 
\[
\Pr[Bad_2(i) \mid Good_1] ~\le~ \frac{(r!)^2}{(2r)!} ~<~ 2^{-r}\;.
\]
\end{lemma}

\begin{proof} 
Let $T$ be a randomly generated computation tree,
let $P = P_S(T) = \langle(e_1,d_1),(e_2,d_2),\ldots,(e_k,d_k)\rangle$,
and let $u_0,\dots,u_{k+1}$ be the nodes on $P$ in $T$.

The event $Bad_2(i)$ occurs if $\sigma_{P}(\aaa^1_{i,j}) \le k$, for
every $j \in [r]$, and $\sigma_{P}(\bbb^1_{i}) = \infty$. Hence, 
when $Bad_2(i)$ occurs we may assume that there exists indices $\ell_j
= \sigma_{P}(\aaa^1_{i,j}) \le k$ for all $j \in [r]$.
For all $\ell_j$ we have $\bbb^1_i,\aaa^1_{i,j} \subseteq
F(u_{\ell_j})$. Hence, Lemma \ref{lemma:policies} shows that the index
$\ell_j$ must be part of a $\bbb$-phase, since the edge
$e_{\ell_j}$ would otherwise be part of the current tree $B(u_{\ell_j})$.
It follows from Lemma \ref{lemma:policies} that $\bbb^1_i \cap
B(u_{\ell_j}) = \emptyset$ for all $j \in [r]$. Since $|\bbb^1_i| =
rs$ and $|\aaa^1_{i,j}| = s$, for all $j\in[r]$,
it follows that the $w$-th time, for $w \in [r]$, we pick an
edge from a new set $\aaa^1_{i,j}$ without picking an edge from
$\bbb^1_i$, this happens with probability at most
$\frac{(r-w+1)s}{rs + (r-w+1)s} = \frac{r-w+1}{2r-w+1}$. 
Note that conditioning on $Good_1$ does not affect this probability.
Hence,
\[
\Pr[Bad_2(i) \mid Good_1] ~\le~ \prod_{w=1}^r \frac{r-w+1}{2r-w+1} ~=~
\prod_{w=1}^r \frac{w}{r+w} ~=~ \frac{(r!)^2}{(2r)!} ~<~ 2^{-r}\;.
\]
The last inequality follows from a simple proof by induction.
\end{proof}

Recall that $|S| = p$.

\begin{lemma}\label{lemma: bad 1} For every multi-edge $\eee\in \mathcal{M}$, we
have $\Pr[Bad_3(\eee) \mid Good_1] \le n^{-2}$, assuming
  $s=2p(r+1)+t$ and $t = 15 \lceil \log n \rceil = \Theta(\log n)$.
\end{lemma}

\begin{proof}
Let $T$ be a randomly generated computation tree,
let $P = P_S(T) = \langle(e_0,d_0),(e_1,d_1),\ldots,(e_k,d_k)\rangle$,
and let $u_0,\dots,u_{k+1}$ be the nodes on $P$ in $T$.

Define the set of sets of edges:
\[
\mathcal{A} ~:=~ \{\aaa^1_{i_q,j} \mid q \in [p], j \in [r]\} \cup
\{\bbb^1_{i_q} \mid q \in [p]\}\;.
\]
Observe that if an edge from each of the sets in $\mathcal{A}$ is
picked along $P$, 
then $P$ is  $(i_1,i_2,\ldots,i_p)$-canonical and we do not get the
event $Bad_3(\eee)$. Observe also that $|\mathcal{A}| = p(r+1)$, and
that each of the sets in $\mathcal{A}$ contains at least $s$
edges. Furthermore, Lemma \ref{lemma:policies} shows that for every
index $\ell \in \{0,\dots,k\}$ there exists a set $A \in \mathcal{A}$
such that $\sigma_P(A) \ge \ell$, $A \subseteq F(u_\ell)$, and $A \cap
B(u_\ell) = \emptyset$. Indeed, in every $\bbb_q$-phase the set
$\bbb^1_{i_q}$ has this property, and in every $\aaa_q$-phase there
exists some $\aaa^1_{i_q,j}$, for $j\in [r]$, with this property.
Hence, there is always a set in $\mathcal{A}$ from which no edge has
previously been picked and for which all edges are available to be picked.

The event $Bad_3(\eee)$ occurs only if we pick $t$ edges from $\eee$
before picking edges $p(r+1)$ times from new sets from
$\mathcal{A}$. Every time we either pick an edge from $\eee$ or from a
new set from $\mathcal{A}$, the edge is picked from $\eee$ with
probability at most $t/(s+t)$. 
Note that conditioning on $Good_1$ does not affect this probability.
Let $X$ be the sum of $p(r+1)+t$ Bernoulli
trials, each with a `success' probability of $p=t/(s+t)$. $Bad_3(\eee)$
occurs only if $X \ge t$. We now apply
Chernoff. $\mu=\E{X}=\frac{t}{s+t}(p(r+1)+t)$. As
$s=2p(r+1)+t$, we then have $\mu=\frac{t}{2}$. Then, 
\[ 
\Pr[X \ge t] ~=~ \Pr[X \ge 2\mu] ~\le~ \left(\frac{{\rm e}}{4}\right)^\mu ~=~
\left(\frac{{\rm e}}{4}\right)^{t/2} ~<~ n^{-2}.
\] 
It is possible that there are fewer than $p(r+1)+t$ trials, which only
lowers the probability.
\end{proof}

We next prove Lemma~\ref{L-appear-S} by showing that:
\[
\PPr{ Good_1 }\, ( 1 - \PPr{ Bad_2 \mid Good_1 } - \PPr{ Bad_3 \mid
  Good_1})
~\ge~ \frac{1}{2\, p!} \;,
\]
where $p = \lfloor \sqrt{n}\rfloor$. In order for the proof to work we pick
$r=\log (4n) = \Theta(\log n)$, $t=15 \lceil \log n 
\rceil=\Theta(\log n)$, and, to satisfy
the assumption for Lemma~\ref{lemma: bad 1},
$s=2p(r+1)+t = \Theta(\sqrt{n}\log n)$.

\medskip\noindent
\textbf{Proof of Lemma \ref{L-appear-S}:\ }
Observe first that the events $Bad_2(i)$ and $Bad_2(i')$ are disjoint for
$i \ne i'$. Using Lemma \ref{lemma:bad_2}, it follows that:
\[
\PPr{ Bad_2 \mid Good_1} ~=~ \sum_{i \in [n] } \PPr{ Bad_2(i) \mid
  Good_1} ~\le~ n\,2^{-r} \le \frac{1}{4}\;.
\]

Similarly, the events $Bad_3(\eee)$ and $Bad_3(\eee')$ are
disjoint for $\eee \ne \eee'$.
The number of multi-edges is $|\mathcal{M}| = 
n(2rs + r + 3) = O(n^{3/2}\log^2 n)$. We, thus, get from
Lemma~\ref{lemma: bad 1} that:
\[
\PPr{ Bad_3 \mid Good_1} ~=~ \sum_{\eee \in \mathcal{M} } \PPr{ Bad_3(\eee) \mid
  Good_1} ~\le~ |\mathcal{M}|\,n^{-2} \le \frac{1}{4}\;.
\]

Using Lemma~\ref{lemma: good 3} it follows that:
\begin{align*}
\sum_{P \in canon_S(G)} \PPr{P \in T} &~=~
\PPr{ Good_1 }\, ( 1 - \PPr{ Bad_2 \mid Good_1 } - \PPr{ Bad_3 \mid
  Good_1})\\
&~\ge~ 
\frac{1}{p!} \left(1-\frac{1}{4}-\frac{1}{4}\right) ~=~ \frac{1}{2\,p!}\;.
\end{align*}
\qed

%% file: Lower-RandomFacet-star.tex
\section{Lower bound for $\RandomPerm$}\label{S-one-pass-variant}

We next prove a lower bound for the $\RandomPerm$
algorithm. Recall that $\RandomPerm$  receives, as a third
argument, a \emph{permutation function} $\sigma: E \to \NN$ that
assigns to each edge
of $E$ a \emph{distinct} natural number. See
Figure~\ref{F-RandomFacet} for a precise description of this
variant of the $\RandomFacet$ algorithm.
Instead of choosing a \emph{random}
edge~$e$ from $F \setminus B$, $\RandomPerm(F,B,\sigma)$
chooses the edge $e \in F \setminus \sigma$ with the \emph{smallest
  permutation index} $\sigma(e)$.

Let $G_{n, r, s, t} = (V,E,c)$ be defined as in Section~\ref{sec:construction}.
$n \in \NN$ is the number of bits
of the corresponding binary counter, and $r,s,t \in \NN$ are
parameters that will be specified by the analysis
to ensure that the expected number
of improving switches performed by the
$\RandomPerm$ algorithm is at least ${\rm e}^{\Omega(\sqrt{n})}$. We
show that it suffices to use $r = s =
t = \Theta(\log n)$.
Since the graph $G_{n,r,s,t}$ contains
$N=\Theta(nrs)=\Theta(n\log^2 n)$ vertices and
$M=\Theta(nrst)=\Theta(n\log^3 n)$ edges, we prove the following
theorem.

\begin{theorem}\label{thm:main1pi}
The expected number of switches performed by the
$\RandomPerm$ algorithm, when run on the graph $G_{n,r,s,t}$, where
$r = s = t = \Theta(\log n)$, with initial tree $B_0$, is at least ${\rm
  e}^{\Omega(\sqrt{n})} = {\rm e}^{\Omega(\sqrt{M}/\log^{3/2} M)}$,
where $M$ is the number of edges in $G_{n,r,s,t}$.
\end{theorem}

We let $\Sigma$ be the set of all permutation functions for
$G_{n,r,s,t} = (V,E,w)$ with range $\{1,2,\ldots,|E|\}$.
As in Section \ref{S-lower-bound} it is helpful to define notation
that allows us to identify interesting events.
Let $\sigma \in \Sigma$ be a permutation function.
Recall that $\mathcal{M}$ is the set of multi-edges, with $\eee \in
\mathcal{M}$ being a set $\eee = \{e^1,e^2,\dots,e^t\}$ of $t$
identical edges. Define:
\[
\renewcommand{\arraystretch}{1.5}
\begin{array}{rrcl}
\forall i \in [n]: &
\sigma(\bbb^1_{i}) &\!=\!& \displaystyle\min_{j\in[rs]} ~ \sigma(b_{i,j}^1) \\
\forall i \in [n], \forall j \in [r]: &
\sigma(\aaa^1_{i,j}) &\!=\!& \displaystyle\min_{k\in[s]} ~ \sigma(a_{i,j,k}^1)\\
\forall i \in [n]: &
\sigma(\aaa^1_i) &\!=\!& \displaystyle\max_{j \in [r]} ~
\displaystyle\min_{k\in[s]} ~ \sigma(a_{i,j,k}^1) \\
\forall \eee \in \mathcal{M}: &
\sigma(\eee) &\!=\!& \displaystyle\max_{e \in \eee} ~ \sigma(e)
\end{array}
\]

\begin{definition}[\bf Induced permutation function]\label{def:induced}
For every permutation function $\sigma \in \Sigma$ we define the
\emph{induced permutation function}, $\hat{\sigma}: [n] \to [n]$, for
$\sigma$ to be the (unique) permutation function that satisfies
$\hat{\sigma}(i) < \hat{\sigma}(j)$ if and only if $\sigma(\bbb^1_i) <
\sigma(\bbb^1_j)$, for all $i,j \in [n]$.
\end{definition}

\begin{definition}[\bf Well-behaved permutation function]\label{def:good}
We say that $\sigma \in \Sigma$ is \emph{well-behaved} if and only if:
\begin{itemize}
\item[$(i)$]
$\sigma(\bbb^1_i) < \sigma(\aaa^1_i)$ for all $i \in [n]$.
\item[$(ii)$]
$\sigma(\aaa^1_{i,j}) < \sigma(\eee)$ for all $\eee \in \mathcal{M}$, $i \in
  [n]$, and $j\in [r]$.
\end{itemize}
\end{definition}


The requirements $(i)$ and $(ii)$ of Definition~\ref{def:good}
correspond to the requirements $(ii)$ and $(iii)$, respectively, of
Definition \ref{D-canonical}. Recall that the one-permutation randomized
counter, $\RandCount^{1P}$, takes as its second argument a
permutation function $\hat\sigma: [n] \to [n]$. Recall also that
$f^{1P}([n],\hat\sigma)$ is
the number of times $\RandCount^{1P}([n],\hat\sigma)$ sets a bit to 1, and
that $f^{1P}(n)$ is the expected value of $f^{1P}([n],\hat\sigma)$
when $\hat\sigma$ is
uniformly random. Let $g^{1P}(F,B,\sigma)$ be the number
of improving switches performed by
$\RandomPerm(F,B,\sigma)$, and let $g^{1P}(F,B)$ be the
expected number of improving switches performed by
$\RandomPerm(F,B,\sigma)$ when $\sigma \in \Sigma$ is picked
uniformly at random.

The following lemma shows that $\RandomPerm$ essentially
simulates $\RandCount^{1P}$ when run
on $G_{n,r,s,t}$. In particular, for $F=E$ and $p = n+1$, the
condition described by $(i)$ is satisfied for the initial
tree $B_0$, such that the lemma says that
$g^{1P}(E,B_0,\sigma) \ge f^{1P}([n],\hat\sigma)$.
The proof of the lemma is
very similar to the proof of Lemma \ref{lemma:policies}. Note also
that the two lemmas describe the same kind of trees.

\begin{lemma}\label{lemma:technical_star}
Let $\sigma \in \Sigma$ be a well-behaved permutation function, and
let $F \subseteq E$ satisfy $reset(F) = 0$. Furthermore, assume that
for all $\eee \in \mathcal{M}$ we have $e \in F$ where $\sigma(e) =
\sigma(\eee)$. In particular, $F$ is functional.
Let $p \in [n+1]$, define $N(F,p) := \{i \in [n] \mid i
< p \land \bbb^1_i \subseteq F\}$, and let $B$ be a tree. Assume
that one of the following two conditions are satisfied:
\begin{itemize}
\item[$(i)$]
\begin{itemize}
\item
For all $i > p$: $bit_i(F,B) = 1$.
\item
$\bbb^1_{p} \subseteq B$, and $\aaa^1_{i,j}\subseteq F$ where
$\sigma(\aaa^1_{i,j}) = \sigma(\aaa^1_{i})$.
\item
For all $i < p$: $bit_i(F,B) = 0$. Furthermore, if $\bbb^1_i
\subseteq F$ then $\aaa^1_{i,j}\subseteq F$ where
$\sigma(\aaa^1_{i,j}) = \sigma(\aaa^1_{i})$.
\end{itemize}
\item[$(ii)$]
\begin{itemize}
\item
For all $i \ne p$: $bit_i(F,B) = 1$.
\item
$\bbb^1_p \subseteq F$, $\aaa^1_{p,j} \subseteq F$ where
  $\sigma(\aaa^1_{p,j}) = \sigma(\aaa^1_{p})$, and $\aaa^1_{p,j}
  \cap B = \emptyset$ for all $j \in [r]$. Furthermore, 
$e \in B$ where $\sigma(e) = \sigma(\bbb^1_p)$.
\end{itemize}
\end{itemize}
Then $g^{1P}(F,B,\sigma) \ge f^{1P}(N(F,p),\hat\sigma)$.
\end{lemma}

\begin{proof}
The lemma is proved by induction in $|N(F,p)|$, $p$, and $|F|$, and backward
induction in $|\bbb^1_p \cap B|$ and $|\aaa^1_{p,j} \cap B|$,
where $\sigma(\aaa^1_{p,j}) = \sigma(\aaa^1_{p})$.
For $|N(F,p)| = 0$ or $p = 1$ we have $N(F,p) = \emptyset$, and the
lemma is clearly true since $f^{1P}(\emptyset,\hat\sigma) = 0$.

Assume that $|N(F,p)| > 0$, and let $e = \argmin_{e \in F\setminus
  B} \sigma(e)$. We consider four cases:
\begin{enumerate}
\item
$\sigma(e) = \sigma(\bbb^1_{i'})$ for some ${i'} \in [n]$ with $\bbb^1_{i'}
  \subseteq F$.
\item
$e \in \bbb^1_{i'}$ for some ${i'} \in [n]$ with $\bbb^1_{i'} \subseteq F$, and
  $\sigma(e) \ne \sigma(\bbb^1_{i'})$.
\item
$e \in \aaa^1_{i',j}$ for some $i'\in[n]$ and $j\in[r]$ such that
$\bbb^1_{i'} \subseteq F$, $\aaa^1_{i'} \sqsubseteq F$, and
  $\aaa^1_{i'} \not\sqsubseteq F\setminus \{e\}$.
\item
$e$ does not qualify for any of the other cases.
\end{enumerate}

\noindent
\textbf{Case 1:}
Assume that $\sigma(e) = \sigma(\bbb^1_{i'})$ for some ${i'} \in [n]$ where
$\bbb^1_{i'} \subseteq F$. 
Let $B'$ be the tree returned by
$\RandomPerm(F\setminus\{e\},B,\sigma)$. 
Since $F\setminus \{e\}$ is functional we know from Lemma
\ref{L-Sigma-optimal} that $B' \subseteq \BB_F$. Furthermore, since
$reset(F\setminus \{e\}) = 0$ we get from Lemma \ref{L-make_switch}
that $e$ is an improving switch with respect to $B'$. Let $B'' =
B'[e]$. It follows that:
\[
g^{1P}(F,B,\sigma) ~=~ g^{1P}(F\setminus\{e\},B,\sigma) + 1 +
g^{1P}(F,B'',\sigma)\;.
\]
Let $i \in \argmin_{i \in N(F,p)} \hat\sigma(i)$. Then we also have:
\[
f^{1P}(N(F,p),\hat\sigma) ~=~ f^{1P}(N(F,p)\setminus \{i\},\hat\sigma)
+ 1 + f^{1P}(N(F,p)\cap [i-1],\hat\sigma)\;.
\]
Using the induction hypothesis, we show that
$g^{1P}(F\setminus\{e\},B,\sigma) \ge 
f^{1P}(N(F,p)\setminus \{i\},\hat\sigma)$ and $g^{1P}(F,B'',\sigma)
\ge f^{1P}(N(F,p)\cap [i-1],\hat\sigma)$ which proves the induction step.

We next prove that $F$, $B$, and $\sigma$ satisfy
$(i)$. Indeed, for $(ii)$ we have $bit_i(F,B) = 1$ for all $i \ne
p$, and if $\sigma(e) = \sigma(\bbb^1_p)$ then $e \in B$. It follows
that every edge $e$ with $\sigma(e) = \sigma(\bbb^1_i)$ is 
part of the current tree in $(ii)$, such that $e \not\in F\setminus
  B$. A similar argument for $(i)$ shows that $i' < p$. 
Since $\bbb^1_{i'} \subseteq F$ it
follows that $i' \in N(F,p)$. Moreover, $i' \in \argmin_{i \in N(F,p)}
\sigma(\bbb^1_i) = \argmin_{i \in N(F,p)} \hat\sigma(i)$. Hence, $i'$
is also the index picked by
$\RandCount^{1P}(N(F,p),\hat\sigma)$. Observe that $N(F\setminus
\{e\},p) = N(F,p) \setminus \{i'\}$, that $F\setminus \{e\}$ is
functional, that $reset(F\setminus \{e\}) = 0$, and that $F\setminus \{e\}$,
$B$, and $\sigma$ satisfy $(i)$. Hence, for the first recursive
call, $\RandomPerm(F\setminus\{e\},B,\sigma)$, it follows by
induction that $g^{1P}(F\setminus\{e\},B,\sigma) \ge f^{1P}(N(F,p)
\setminus \{i'\},\hat\sigma)$.

Lemma \ref{L-Sigma-optimal} shows that $F$, $B''$, and $\sigma$ satisfy
$(ii)$ where $i'$ plays the role of $p$. Since $i' < p$ it follows by
induction that $g^{1P}(F,B'',\sigma)
\ge f^{1P}(N(F,p)\cap [i'-1],\hat\sigma)$.

\medskip\noindent
\textbf{Case 2:}
Assume that $e \in \bbb^1_{i'}$ for some ${i'} \in [n]$ with
$\bbb^1_{i'} \subseteq F$, and $\sigma(e) \ne \sigma(\bbb^1_{i'})$.
Let $B'$ be the tree returned by
$\RandomPerm(F\setminus\{e\},B,\sigma)$, and let $B'' =
B'[e]$. As in case 1 we have:
\begin{align*}
g^{1P}(F,B,\sigma) ~=~ g^{1P}(F\setminus\{e\},B,\sigma) + 1 +
g^{1P}(F,B'',\sigma)\;.
\end{align*}
Using the induction hypothesis, we show that
$g^{1P}(F,B'',\sigma) \ge f^{1P}(N(F,p),\hat\sigma)$ which proves
the induction step, i.e., in this case we only count the improving
switches performed during the second recursive call.

We first show that $F$, $B$, and $\sigma$ satisfy
$(ii)$. Indeed, for $(i)$ we have $\bbb^1_i \subseteq B$ for all $i \ge
p$ with $\bbb^1_{i} \subseteq F$, which implies that $i' < p$. On the
other hand, $\bbb^1_{i} \cap B = \emptyset$ for all $i < p$. Since
$\bbb^1_{i'} \subseteq F$ it follows that $\sigma(e) =
\sigma(\bbb^1_{i'})$ which is a contradiction. Moreover, since 
$F$, $B$, and $\sigma$ satisfy $(ii)$ such that $bit_i(F,B) = 1$
for all $i \ne p$, we must have $i' = p$.
Lemma \ref{L-Sigma-optimal} then shows that $F$, $B''$, and $\sigma$ satisfy
$(ii)$. Furthermore, we have $|\bbb^1_p \cap B''| > |\bbb^1_p \cap
B|$, and we get by induction that $g^{1P}(F,B'',\sigma)
\ge f^{1P}(N(F,p),\hat\sigma)$. Note that if $|\bbb^1_p \cap B| =
rs$ such that $\bbb^1_p \subseteq B$, then it was not possible to
pick $e \in \bbb^1_{p}$, and case 2 could not happen. The base-case
$|\bbb^1_p \cap B| = rs$ then follows from the proof of the other
three cases.

\medskip\noindent
\textbf{Case 3:}
Assume that $e \in \aaa^1_{i',j}$ for some $i'\in[n]$ and $j\in[r]$ such
that $\bbb^1_{i'}$, $\aaa^1_{i'} \sqsubseteq F$, and $\aaa^1_{i'}
\not\sqsubseteq F\setminus \{e\}$. The most important observation for
this case is that $reset(F\setminus \{e\}) = i'$.
Let $B'$ be the tree returned by
$\RandomPerm(F\setminus\{e\},B,\sigma)$.
Since $F\setminus \{e\}$ is functional we know from Lemma
\ref{L-Sigma-optimal} that $B' \subseteq \BB_F$. Furthermore, since
$reset(F\setminus \{e\}) = i'$ we get from Lemma \ref{L-make_switch}
that $e$ is an improving switch with respect to $B'$.
Let $B'' = B'[e]$. 
As in cases 1 and 2 we therefore have:
\begin{align*}
g^{1P}(F,B,\sigma) ~=~ g^{1P}(F\setminus\{e\},B,\sigma) + 1 +
g^{1P}(F,B'',\sigma)\;.
\end{align*}
Using the induction hypothesis, we show that
$g^{1P}(F,B'',\sigma) \ge f^{1P}(N(F,p),\hat\sigma)$ which proves
the induction step, i.e., as in case 2 we again only count the improving
switches performed during the second recursive call.

We next show that for both $(i)$ and $(ii)$ we must have $i' = p$.
Consider first the case when $F$, $B$, and $\sigma$ satisfy
$(i)$. Since $bit_i(F,B) = 1$ for all $i > p$ and $reset(F) = 0$, we
must have $\aaa^1_i \sqsubseteq B$ if $\bbb^1_i \subseteq B$. It
follows that we can not have $\aaa^1_{i'} \not\sqsubseteq F\setminus
\{e\}$ if $i' > p$. Suppose that $i' < p$. Since $\bbb^1_{i'}
\subseteq F$ we have $\aaa^1_{i',j}\subseteq F$ where 
$\sigma(\aaa^1_{i',j}) = \sigma(\aaa^1_{i'})$. Hence, we must have $e
\in \aaa^1_{i',j}$, which means that $\sigma(e) \ge \sigma(\aaa^1_{i'})$.
Since $\sigma$ is well-behaved it satisfies
$\sigma(\bbb^1_{i'}) < \sigma(\aaa^1_{i'}) \le \sigma(e)$.
On the other hand, for all $i < p$ we have $\bbb^1_i \cap B =
\emptyset$. Hence, the edge $e'$ where $\sigma(e') =
\sigma(\bbb^1_{i'})$ is available and would have been picked instead
of $e$; a contradiction. For $(ii)$ we get that $i'=p$ from the fact
that $reset(F) = 0$ and $bit_i(F,B) = 1$ for all $i \ne p$.

Lemma \ref{L-Sigma-optimal} then shows that $F$, $B''$, and $\sigma$ satisfy
$(i)$. Note that for both $(i)$ and $(ii)$ we must have $e \in
\aaa^1_{p,j}$ where $\sigma(\aaa^1_{p,j}) = \sigma(\aaa^1_{p})$.
Hence, $|\aaa^1_{p,j} \cap B''| > |\aaa^1_{p,j} \cap
B|$, and we get by induction that $g^{1P}(F,B'',\sigma)
\ge f^{1P}(N(F,p),\hat\sigma)$. Note that if $|\aaa^1_{p,j} \cap B| =
s$ such that $\aaa^1_{p,j} \subseteq B$, then it was not possible to
pick $e \in \aaa^1_{p,j}$, and case 3 could not happen. The base-case
$|\aaa^1_{p,j} \cap B| = s$ then follows from the proof of the other
three cases.

\medskip\noindent
\textbf{Case 4:}
Assume that $e \not\in \bbb^1_i$ for some $i$ where $\bbb^1_i
\subseteq F$, and that 
$e \not\in \aaa^1_{i,j}$ for some $i\in[n]$ and $j\in[r]$ such that
$\bbb^1_{i} \subseteq F$, $\aaa^1_{i} \sqsubseteq F$, and
  $\aaa^1_{i} \not\sqsubseteq F\setminus \{e\}$.
We show that $g^{1P}(F\setminus\{e\},B,\sigma) \ge
f^{1P}(N(F,p),\hat\sigma)$, i.e., in this case we only count the improving
switches performed during the first recursive call.

Observe that $N(F\setminus\{e\},p) = N(F,p)$ since $e \not\in
\bbb^1_i$ for some $i$ where $\bbb^1_i \subseteq F$. Also observe that
$reset(F\setminus \{e\}) = 0$ since $e \not\in \aaa^1_{i,j}$ for some
$i\in[n]$ and $j\in[r]$ such that $\bbb^1_{i} \subseteq F$,
$\aaa^1_{i} \sqsubseteq F$, and $\aaa^1_{i} \not\sqsubseteq F\setminus
\{e\}$. For the same reason we see that for all $i\in[n]$, if
$\bbb^1_i \subseteq F \setminus\{e\}$ then
$\aaa^1_{p,j} \subseteq F\setminus\{e\}$ where
  $\sigma(\aaa^1_{p,j}) = \sigma(\aaa^1_{p})$. It follows that if $F$,
$B$, and $\sigma$ satisfy $(i)$ or $(ii)$, respectively, then
$F\setminus \{e\}$, $B$, and $\sigma$ satisfy $(i)$ or $(ii)$,
correspondingly. 
It remains to show that
for all $\eee \in \mathcal{M}$ we have $e' \in F\setminus\{e\}$ where
$\sigma(e') = \sigma(\eee)$. Once this has been shown it follows by
induction that $g^{1P}(F\setminus\{e\},B,\sigma) \ge
f^{1P}(N(F,p),\hat\sigma)$, since $|F \setminus \{e\}| < |F|$. Note
that it is not possible to reach a situation where $F \cap \eee =
\emptyset$ for some $\eee \in \mathcal{M}$. Hence, the induction
always goes back to one of the other base-cases.

Suppose that $e \in \eee$ for some $\eee \in \mathcal{M}$, and assume
for the sake of contradiction that $\eee \cap (F\setminus \{e\}) =
\emptyset$. This is only possible if $\sigma(e) = \sigma(\eee)$. Since
$p > 1$ and $N(F,p) \ne \emptyset$ there exists some set
$\aaa^1_{i,j}$ such that $\aaa^1_{i,j} 
\subseteq F$ and $\aaa^1_{i,j} \cap B = \emptyset$. For $(i)$ this
is true for some $i<p$, and for $(ii)$ it is true for $i = p$. Since
$\sigma$ is well-behaved it must be the case that
$\sigma(\aaa^1_{i,j}) < \sigma(\eee)$. Hence, an edge from
$\aaa^1_{i,j}$ would be picked before an edge from $\eee$, and we get
a contradiction.
\end{proof}

We next prove that a uniformly random permutation function
$\sigma \in \Sigma$ is well-behaved with high probability. The proof
of the following lemma resembles that of Lemma \ref{lemma:bad_2}.

\begin{lemma} \label{lemma: probability of good permutation}
Let $\sigma \in \Sigma$ be chosen uniformly at
random. If $r=s=t=3\lceil \log n\rceil$, then $\sigma$ is well-behaved
with probability at least $1/2$.
\end{lemma}

\begin{proof}
We start by upper bounding the probability that $\sigma$ fails to
satisfy Definition~\ref{def:good} $(i)$. Let $i \in [n]$ be given, and
define:
\[
\renewcommand{\arraystretch}{1.5}
\begin{array}{rrcl}
\forall j \in [r]: & \bbb^1_{i,j} &\!=\!& \{b^1_{i,1+(j-1)s}, b^1_{i,2+(j-1)s},
\dots, b^1_{i,s+(j-1)s}\}\\
\forall j \in [r]: & \sigma(\bbb^1_{i,j}) &\!=\!& \min_{e \in
  \bbb^1_{i,j}} \sigma(e)\\
\end{array}
\]
Note that the sets $\bbb^1_{i,j}$, for $j\in [r]$, partition $\bbb^1_i$ into
$r$ sets of size $s$. In particular, all the sets $\bbb^1_{i,j}$ and
$\aaa^1_{i,j}$, for $j \in [r]$, have the same size. Hence, we may view
$\sigma$ as defining a uniformly random permutation $\sigma'$ of the set
$\mathcal{S} =
\{\bbb^1_{i,1},\dots,\bbb^1_{i,r},\aaa^1_{i,1},\dots,\aaa^1_{i,r}\}$,
such that for all $s_1,s_2 \in \mathcal{S}$, $\sigma'(s_1) <
\sigma'(s_2)$ if and only if $\sigma(s_1) < \sigma(s_2)$.
In order to get the event $\sigma(\bbb^1_{i}) > \sigma(\aaa^1_{i})$ it
must be the case that all the $\aaa^1_{i,j}$ elements of $\mathcal{S}$
are placed before all the $\bbb^1_{i,j}$ elements of $\mathcal{S}$. This
happens with probability:
\[
\PPr{\sigma(\bbb^1_{i}) > \sigma(\aaa^1_{i})} ~=~ \frac{(r!)^2}{(2r)!}
\]
Hence, the probability that $\sigma$ fails to satisfy
Definition~\ref{def:good} $(i)$ is at most
$n\frac{(r!)^2}{(2r)!} \le n 2^{-r} \le 1/4$, where the first
inequality follows from a simple proof by induction.

We next upper bound the probability that $\sigma$ fails to satisfy
Definition~\ref{def:good} $(ii)$.
Let $\eee \in \mathcal{M}$, $i \in [n]$, and $j \in [r]$ be
given. Recall that $\eee$ contains $t$ copies of a
multi-edge. Similarly, $\aaa^1_{i,j}$ contains $s$ edges. The event
$\sigma(\aaa^1_{i,j}) > \sigma(\eee)$ occurs if and only if all the
edges of $\aaa^1_{i,j}$ are assigned higher indices than all the edges
of $\eee$. Since $\sigma$ is uniformly random this happens with
probability:
\[
\PPr{\sigma(\aaa^1_{i,j}) > \sigma(\eee)} ~=~ \frac{s!\,t!}{(s+t)!}
\]
Since $|\mathcal{M}| = n(2rs + r + 3)$, it follows that $\sigma$ fails to
satisfy Definition~\ref{def:good} $(ii)$ with probability at most
$n^2r(2rs + r + 3) \frac{s!\,t!}{(s+t)!} \le 6n^2\log^3 n \,2^{-3\log n}
\le 1/4$.

Let $p_{n,r,s,t}$ be the probability that $\sigma$ is well-behaved. We
have shown that:
\[
p_{n,r,s,t} ~\geq~ 1 - n\frac{(r!)^2}{(2r)!} - n^2r(2rs + r + 3)
\frac{s!\,t!}{(s+t)!} ~\ge~ \frac{1}{2}\;.
\]
\end{proof}

\medskip\noindent
\textbf{Proof of Theorem \ref{thm:main1pi}:\ }
If $\sigma \in \Sigma$ is picked uniformly at random then we know
from Lemma~\ref{lemma: probability of good permutation} that $\sigma$
is well-behaved with probability at least $1/2$. Moreover, since every
set $\bbb^1_i$, for $i \in [n]$, has the same size, the induced
permutation function $\hat\sigma$ is also uniformly random.
It then follows from Lemma \ref{lemma:counters} and Lemma
\ref{lemma:asymp} that:
\[
g^{1P}(F,B) ~\ge~ \frac{1}{2}f^{1P}(n) ~=~ \frac{1}{2}f(n)
~\sim~ \frac{{\rm e}^{2\sqrt{n}}}{4\sqrt{\pi{\rm e}}\, n^{1/4}}
~=~ {\rm e}^{\Omega(\sqrt{n})}\;.
\]

\qed

%% file: Lower-RandomBland.tex
\section{Lower bound for \RandomBland}\label{S-randomized-blands-rule}

We next prove a lower bound for the $\RandomBland$
algorithm. The proof of this lower bound is very similar to the
corresponding proof for $\RandomPerm$. In fact, the only difference is
that Lemma \ref{lemma:technical_star} needs to be proved slightly
differently. Recall that $\RandomBland$ is the algorithm obtained by running $\Bland$ with a uniformly random permutation. Since, we remove edges and consider subproblems, it is again convenient to operate with a \emph{permutation function} $\sigma: E \to \NN$
that assigns to each edge
of $E$ a \emph{distinct} natural number.
Recall that $\Sigma$ is the
set of permutation functions with range $[|E|]$. The algorithm works by
repeatedly performing the improving switch with \emph{largest} permutation
index $\sigma(e)$. We study the recursive formulation of
$\Bland$ shown in Figure~\ref{F-RandomFacet}. $\Bland$ is
different from $\RandomPerm$ in the following ways. The edge chosen to
be removed from $F$ is always the edge in $F$ with \emph{smallest}
permutation index $\sigma(e)$, regardless of whether $e\in B$ or
not. This means that it is possible to have $B \not\subseteq F$,
which means that the termination criterion is changed to $F =
\emptyset$.

Let $G_{n, r, s, t} = (V,E,c)$ be defined as in Section~\ref{sec:construction}.
The parameters $n,r,s,t \in \NN$ play the same role as in Section
\ref{S-one-pass-variant}. In particular, we use $r = s =
t = \Theta(\log n)$, such that $G_{n,r,s,t}$ contains
$N=\Theta(nrs)=\Theta(n\log^2 n)$ vertices and
$M=\Theta(nrst)=\Theta(n\log^3 n)$ edges. We also use the same
notation as in Section \ref{S-one-pass-variant}. The following theorem
is again proved by showing that $\Bland(E,B_0,\sigma)$
simulates $\RandCount^{1P}([n],\hat\sigma)$, where $\hat\sigma$ is the
induced permutation function, when $\sigma$
is well-behaved. Recall that Lemma 
\ref{lemma: probability of good permutation} shows that $\sigma$ is
well-behaved with probability at least $1/2$.

\begin{theorem}\label{thm:mainbland}
The expected number of switches performed by the
$\RandomBland$ algorithm, when run on the graph $G_{n,r,s,t}$, where
$r = s = t = \Theta(\log n)$, with initial tree $B_0$, is at least ${\rm
  e}^{\Omega(\sqrt{n})} = {\rm e}^{\Omega(\sqrt{M}/\log^{3/2} M)}$,
where $M$ is the number of edges in $G_{n,r,s,t}$.
\end{theorem}

Let $B'$ be the tree returned by
$\Bland(F,B,\sigma)$. An important difference between
$\Bland$ and $\RandomPerm$ is that $B' \subseteq F \cup B$
is optimal for $G_{F \cup B'}$ for $\Bland$, whereas it is
optimal for $G_F = G_{F \cup B}$ for $\RandomPerm$.
This means that it is harder to keep track of which edges are present
in $F \cup B'$. There are, however, some edges in $B$ that are
guaranteed to also be in $B'$, namely the edges $(u,v)\in B$ used at
all vertices $u$ whose distance to the target $\TT$ in $B$ is optimal for $G_{F\cup
  B}$. Since $B'$ is obtained from $B$ by performing improving
switches such edges must remain in $B'$. This leads us to the
following definition.

\begin{definition}[\bf Fixed edges]
Let $B$ be a tree and $F \subseteq E$. We say that an edge $(u,v)
\in B$ is \emph{fixed} for $B$ with respect to $F$ if $\val_B(u)
= \val(u)$ for the graph $G_{F\cup B}$.
\end{definition}

Recall that $G_{n,r,s,t}$ is acyclic such that it is not possible to
get from higher bits to lower bits. For any edge $(u,v) \in F \subseteq
E$, let $B'$ be returned by
$\Bland(F\setminus\{(u,v)\},B,\sigma)$. Then all edges used in
$B'$ at vertices that can not reach $u$ in $G_{n,r,s,t}$ are fixed
with respect to $F$. We let $V(b_{i',j'}) \subseteq V$ and $V(a_{i',j',k'})
\subseteq V$ be vertices that can not reach $b_{i',j'}$ and
$a_{i',j',k'}$, respectively, including $b_{i',j'}$ and
$a_{i',j',k'}$ themselves:
\begin{align*}
V(b_{i',j'}) ~=~ &\{u_i,w_i \mid i > i'\} \cup 
\{a_{i,j,k} \mid i > i', j\in [r], k \in [s]\}
\cup \{b_{i,j} \mid i > i', j\in [rs]\}
\cup \{b_{i',j} \mid j \ge j'\}\\
V(a_{i',j',k'}) ~=~ &\{u_i,w_i \mid i > i'\} \cup 
\{a_{i,j,k} \mid i > i', j\in [r], k \in [s]\}
\cup \{b_{i,j} \mid i \ge i', j\in [rs]\} \,\cup \\
&\{a_{i',j,k} \mid j \ne j', k \in [s]\} \cup 
\{a_{i',j',k} \mid k \ge k'\}
\end{align*}
We say that $V(b_{i',j'})$ is fixed for $B$ with respect to $F$ if every
edge $(u,v) \in B$, for $u\in V(b_{i',j'})$, is fixed with respect to $F$.
Similarly, we say that $V(a_{i',j',k'})$ is fixed for $B$ with respect to
$F$ if every edge $(u,v) \in B$, for $u\in V(a_{i',j',k'})$, is fixed with
respect to $F$. Note that if $V(a_{i',j',k'})$ is fixed for $B$ with respect to
$F$, for some $i',j',k'$, then $V(b_{i',1})$ is fixed for $B$ with respect to
$F$. Since the choices at $b_{i,j}$ and $a_{i,j,k}$, for some $i,j,k$,
are essentially binary, i.e., the edges in $\bbb^0_{i,j}$ and
$\aaa^0_{i,j,k}$, respectively, are identical, the above discussion
proves the following lemma.

\begin{lemma}\label{lemma:fixed}
Let $B$ be a tree, let $F\subseteq E$, and let $\sigma \in
\Sigma$. Assume that $e = b^1_{i,j} \in F$, for some $i\in[n]$ and $j\in
[r]$, let $B'$ be the tree returned by
$\Bland(F\setminus\{e\},B,\sigma)$, 
assume that $e$ is an improving switch w.r.t. $B'$, and let
$B'' = B'[e]$. Then $V(b_{i,j})$ is fixed for $B''$
w.r.t. $F$. Similarly, if $e = a^1_{i,j,k}$ then $V(a_{i,j,k})$ is fixed
for $B''$ w.r.t. $F$.
\end{lemma}

Let $B'$ be the tree returned by
$\Bland(F,B,\sigma)$. It is also useful to know which edges
are, in a sense, lost when going from $B$ to $B'$, i.e., edges $e$
such that $e \in F\cup B$ but $e \not\in F\cup B'$. The following
lemma will be useful for this purpose.
Let $\sigma^{-1}(\ell)$ be the edge $e$ with index $\ell$, i.e.,
$\sigma(e) = \ell$. 
Define $F(\sigma,\ell) := \{e \in E \mid \sigma(e) \ge \ell\}$. Note that the
set $F$ for $\Bland(F,B,\sigma)$ is always equal to
$F(\sigma,\ell)$ for some $\ell \ge 1$. In particular, $F(\sigma,1) =
E$.

\begin{lemma}\label{lemma:drop}
Let $\sigma\in\Sigma$ be well-behaved, let $\ell \ge 1$, let $F =
F(\sigma,\ell)$, let $B$ be a tree, and let $p \in [n]$. Assume
that $reset(F\cup B) < p$, that $bit_i(F\cup B,B) = 1$ for all $i > p$,
and that $\ell \le\sigma(\aaa^1_p)$. Assume also that there is some
$j' \in [rs]$ such that for all $j \ge j'$ we have 
$b^1_{p,j} \in B$, for all $j < j'$ we have $b^1_{p,j} \in F$, and
$V(b_{p,j'})$ is fixed for $B$ w.r.t. $F$. In particular, 
$\bbb^1_p \subseteq F\cup B$. Finally, let
$B'$ be the tree returned by $\Bland(F,B,\sigma)$. Then for
all $i < p$, $\bbb^1_i \cap B' \subseteq F$.
\end{lemma}

\begin{proof}
Consider the computation path obtained by repeatedly entering the first
branch of the recursion until reaching a recursive call of the
form $\Bland(F(\sigma,\ell'),B,\sigma)$, where either
$\sigma^{-1}(\ell') \in \bbb^1_p\setminus B$ or $\ell' = \sigma(\aaa^1_p)$. 
Note that since $\sigma$ is well-behaved and $\ell' \le
\sigma(\aaa^1_p)$ the set $F(\sigma,\ell'+1)$ is functional. Let 
$B_1'$ be the tree returned by
$\Bland(F(\sigma,\ell'+1),B,\sigma)$. Then 
$B_1'$ is optimal for $G_{F(\sigma,\ell'+1)\cup B_1'}$, and by Lemma
\ref{L-Sigma-optimal} we have $B_1' \subseteq
\BB_{F(\sigma,\ell'+1)\cup B_1'}$. 

Consider the case when $\sigma^{-1}(\ell') \in \bbb^1_p\setminus B$, that is,
$\sigma^{-1}(\ell') = b^1_{p,j}$ for some $j$. Then, since
$reset(F\cup B) < p$ and $bit_i(F,B) = 1$ for all $i > p$, we have
$reset(F(\sigma,\ell'+1)\cup B) < p$, and Lemma \ref{L-make_switch}
shows that $b^1_{p,j}$ is an improving switch
w.r.t. $B_1'$. Moreover, by Lemma \ref{lemma:fixed}, $V(_{p,j})$ is
fixed for $B_1'' = B_1'[b^1_{p,j}]$
w.r.t. $F(\sigma,\ell')$. Furthermore, Lemma \ref{L-Sigma-optimal}
shows that $F(\sigma,\ell')$ and $B_1''$ have the properties that
were assumed about $F$ and $B$ in the lemma. Hence, we may assume,
by induction, that $\ell' = \sigma(\aaa^1_p)$.

Consider next the case when $\ell' = \sigma(\aaa^1_p)$. 
Since $\bbb^1_p \subseteq F(\sigma,\ell'+1)\cup B$ and $V(b_{p,j'+1})$ is
fixed for $B$ w.r.t. $F$, where $j' = \min\{j\in[rs]\mid b^1_{p,j}
\not\in B\}$, it follows that $\bbb^1_p \subseteq
F(\sigma,\ell'+1)\cup B_1'$. On the other hand, we have $\aaa^1_p
\not\sqsubseteq F(\sigma,\ell'+1)\cup B$, which means that
$reset(F(\sigma,\ell'+1)\cup B_1') = p$. It then follows from Lemma
\ref{L-Sigma-optimal} that $\bbb^1_i \cap B_1' = \emptyset$ for all
$i < p$. Since the tree $B'$ returned by
$\Bland(F,B,\sigma)$ is obtained from $B_1'$ by performing
improving switches from $F$ it follows that $\bbb^1_i \cap B'
\subseteq F$ for $i < p$.
\end{proof}

Let $\sigma \in \Sigma$ be a permutation function.
Recall that
$f^{1P}([n],\hat\sigma)$ is
the number of times $\RandCount^{1P}([n],\hat\sigma)$ sets a bit to 1, and
that $f^{1P}(n)$ is the expected value of $f^{1P}([n],\hat\sigma)$
when $\hat\sigma$ is
uniformly random. Let $h^{1P}(F,B,\sigma)$ be the number
of improving switches performed by
$\Bland(F,B,\sigma)$, and let $h^{1P}(F,B)$ be the
expected number of improving switches performed by
$\Bland(F,B,\sigma)$ when $\sigma \in \Sigma$ is picked
uniformly at random, i.e., $h^{1P}(F,B)$ is the
expected number of improving switches performed by
$\RandomBland$.

The following lemma shows that $\Bland$ essentially
simulates $\RandCount^{1P}$ when run
on $G_{n,r,s,t}$. In particular, for $F=E$ and $p = n+1$, the
condition described by $(i)$ is satisfied for the initial
tree $B_0$, such that the lemma says that
$h^{1P}(E,B_0,\sigma) \ge f^{1P}([n],\hat\sigma)$.
The lemma corresponds to Lemma \ref{lemma:technical_star}, and the
proof is very similar. It is used to prove Theorem \ref{thm:mainbland}
in the same way as Lemma \ref{lemma:technical_star} was used to prove
Theorem \ref{thm:main1pi}.

\begin{lemma}\label{lemma:technical_bland}
Let $\sigma \in \Sigma$ be a well-behaved permutation function, let
$B$ be a tree, let
$\ell \ge 1$, let $F = F(\sigma,\ell) \subseteq E$, assume that 
$reset(F\cup B) = 0$, and assume that $F$ is functional.
Let $p \in [n+1]$, and define $N(F,p) := \{i \in [n] \mid i
< p \land \bbb^1_i \subseteq F\}$. Assume
that one of the following two conditions are satisfied:
\begin{itemize}
\item[$(i)$]
\begin{itemize}
\item
For all $i > p$: $bit_i(F\cup B,B) = 1$.
\item
$\bbb^1_{p} \subseteq B$.
\item
Let $j' \in [r]$ satisfy $\sigma(\aaa^1_{p,j'}) =
\sigma(\aaa^1_{p})$. There is some $k' \in [s]$ such that for all $k
\ge k'$ we have $a^1_{p,j',k} \in B$, for all $k < k'$ we have
$a^1_{p,j',k} \in F$, and $V(a_{p,j',k'})$ is fixed for $B$ w.r.t. $F$.
In particular, $\aaa^1_{p,j'} \subseteq F \cup B$.
\item
For all $i < p$: $bit_i(F\cup B,B) = 0$.
\item
$\ell \le \sigma(\bbb^1_p)$.
\end{itemize}
\item[$(ii)$]
\begin{itemize}
\item
For all $i \ne p$: $bit_i(F\cup B,B) = 1$.
\item
There is some $j' \in [rs]$ such that for all $j \ge j'$ we have
$b^1_{p,j} \in B$, for all $j < j'$ we have $b^1_{p,j} \in F$, and
$V(b_{p,j'})$ is fixed for $B$ w.r.t. $F$. In particular, 
$\bbb^1_p \subseteq F\cup B$.
\item
$\ell \le \sigma(\aaa^1_{p})$.
\end{itemize}
\end{itemize}
Then $h^{1P}(F,B,\sigma) \ge f^{1P}(N(F,p),\hat\sigma)$.
\end{lemma}

\begin{proof}
The lemma is proved by induction in $|N(F,p)|$, $p$, and $|F|$, and backward
induction in $|\bbb^1_p \cap B|$ and $|\aaa^1_{p,j} \cap B|$,
where $\sigma(\aaa^1_{p,j}) = \sigma(\aaa^1_{p})$.
For $|N(F,p)| = 0$ or $p = 1$ we have $N(F,p) = \emptyset$, and the
lemma is clearly true since $f^{1P}(\emptyset,\hat\sigma) = 0$.

Assume that $|N(F,p)| > 0$, and let $e = \argmin_{e \in F} \sigma(e) =
\sigma^{-1}(\ell)$. We consider five cases:
\begin{enumerate}
\item
$\sigma(e) = \sigma(\bbb^1_{i'})$ for some ${i'} \in N(F,p)$, and
  $(i)$ is satisfied.
\item
$\sigma(e) = \sigma(\bbb^1_{i'})$ for some ${i'} \in N(F,p)$, and
  $(ii)$ is satisfied.
\item
$e \in \bbb^1_{p}\setminus B$ and $\sigma(e) \ne \sigma(\bbb^1_{p})$.
\item
$e \in \aaa^1_{p,j}\setminus B$ for some $j\in[r]$ such that
$\aaa^1_{p} \not\sqsubseteq (F\setminus \{e\})\cup B$.
\item
$e$ does not qualify for any of the other cases.
\end{enumerate}
For each case we consider the cases where $(i)$ and $(ii)$ are
satisfied separately.

In the following case analysis we let $B'$ be the tree returned by
$\Bland(F\setminus\{e\},B,\sigma)$, and we let $B'' = B'[e]$. 

\noindent
\textbf{Case 1:}
Assume that $\sigma(e) = \sigma(\bbb^1_{i'})$ for some ${i'} \in
N(F,p)$, and assume that $(i)$ is satisfied.
Since $F\setminus \{e\}$ is functional and $B'$ is optimal for
$G_{(F\setminus\{e\})\cup B'}$, we know from Lemma 
\ref{L-Sigma-optimal} that $B' \subseteq \BB_{(F\setminus\{e\})\cup B'}$. 
Since $reset(F\cup B) = 0$ and $V(b_{p,1})$ is fixed for $B$
w.r.t. $F$ we get that 
$\bbb^1_i\subseteq B'$ if $\bbb^1_i \subseteq F\cup B$, for all $i
\ge p$. Similarly, $\aaa^1_i \sqsubseteq B'$ if $\bbb^1_i \subseteq
F\cup B$, for all $i > p$. It follows that
$reset((F\setminus\{e\})\cup B') \ge p$, and we get from Lemma
\ref{L-make_switch} 
that $e$ is an improving switch with respect to $B'$. Let $B'' =
B'[e]$. Hence, we have:
\[
h^{1P}(F,B,\sigma) ~=~ h^{1P}(F\setminus\{e\},B,\sigma) + 1 +
h^{1P}(F,B'',\sigma)\;.
\]
Let $i \in \argmin_{i \in N(F,p)} \hat\sigma(i)$. Then we also have:
\[
f^{1P}(N(F,p),\hat\sigma) ~=~ f^{1P}(N(F,p)\setminus \{i\},\hat\sigma)
+ 1 + f^{1P}(N(F,p)\cap [i-1],\hat\sigma)\;.
\]
Using the induction hypothesis, we show that
$h^{1P}(F\setminus\{e\},B,\sigma) \ge 
f^{1P}(N(F,p)\setminus \{i\},\hat\sigma)$ and $h^{1P}(F,B'',\sigma)
\ge f^{1P}(N(F,p)\cap [i-1],\hat\sigma)$ which proves the induction step.

Observe that
$i' \in \argmin_{i \in N(F,p)}
\sigma(\bbb^1_i) = \argmin_{i \in N(F,p)} \hat\sigma(i)$ such that
$\RandCount^{1P}(N(F,p),\hat\sigma)$ also picks the index $i'$.
Furthermore, observe that $N(F\setminus
\{e\},p) = N(F,p) \setminus \{i'\}$, that $F\setminus \{e\}$ is
functional, that $reset((F\setminus \{e\})\cup B) = 0$, and that
$F\setminus \{e\}$ and 
$B$ satisfy $(i)$. Hence, for the first recursive
call, $\Bland(F\setminus\{e\},B,\sigma)$, it follows by
induction that $h^{1P}(F\setminus\{e\},B,\sigma) \ge f^{1P}(N(F,p)
\setminus \{i'\},\hat\sigma)$.

Let $j'$ be the index such that $e = b^1_{i',j'}$. Note that Lemma
\ref{lemma:fixed} shows that $V(b_{i',j'})$ is fixed for 
$B''$ w.r.t. $F$. Lemma \ref{L-Sigma-optimal} shows that $F$, $B''$, and $\sigma$ satisfy
$(ii)$ where $i'$ plays the role of $p$. Since $i' < p$ it follows by induction that
$h^{1P}(F,B'',\sigma) \ge f^{1P}(N(F,p)\cap [i'-1],\hat\sigma)$.

\medskip\noindent
\textbf{Case 2:}
Assume that $\sigma(e) = \sigma(\bbb^1_{i'})$ for some ${i'} \in
N(F,p)$, and assume that $(ii)$ is satisfied.
By an argument
analogous to the one given in case 1 $(i)$ we see that:
\begin{align*}
h^{1P}(F,B,\sigma) &~=~ h^{1P}(F\setminus\{e\},B,\sigma) + 1 +
h^{1P}(F,B'',\sigma)\\
f^{1P}(N(F,p),\hat\sigma) &~=~ f^{1P}(N(F,p)\setminus \{i\},\hat\sigma)
+ 1 + f^{1P}(N(F,p)\cap [i-1],\hat\sigma)\;,
\end{align*}
where $i \in \argmin_{i \in N(F,p)} \hat\sigma(i)$.
Using the induction hypothesis, we again show that
$h^{1P}(F\setminus\{e\},B,\sigma) \ge 
f^{1P}(N(F,p)\setminus \{i\},\hat\sigma)$ and $h^{1P}(F,B'',\sigma)
\ge f^{1P}(N(F,p)\cap [i-1],\hat\sigma)$ which proves the induction step.

We again observe that $i' \in \argmin_{i \in N(F,p)} \hat\sigma(i)$
is also the index picked by
$\RandCount^{1P}(N(F,p),\hat\sigma)$, that $N(F\setminus
\{e\},p) = N(F,p) \setminus \{i'\}$, that $F\setminus \{e\}$ is
functional, and that $reset((F\setminus \{e\})\cup B) = 0$.
It follows that $F\setminus \{e\}$, 
$B$, and $\sigma$ satisfy $(ii)$. Hence, for the first recursive
call, $\Bland(F\setminus\{e\},B,\sigma)$, it follows by
induction that $h^{1P}(F\setminus\{e\},B,\sigma) \ge f^{1P}(N(F,p)
\setminus \{i'\},\hat\sigma)$.

Note that $\bbb^1_{i'} \not\subseteq F\setminus\{e\}$. Since $i' < p$,
Lemma \ref{lemma:drop} shows that $\bbb^1_{i'} \cap B' \subseteq
F\setminus\{e\}$, and we therefore have $\bbb^1_{i'} \not\subseteq
(F\setminus\{e\}) \cup B'$.
Let $j'$ be the index such that $e = b^1_{i',j'}$. Note also that
Lemma \ref{lemma:fixed} shows that $V(b_{i',j'})$ is fixed for 
$B''$ w.r.t. $F$.
It follows that Lemma \ref{L-Sigma-optimal} shows that $F$, $B''$,
and $\sigma$ satisfy 
$(ii)$ where $i'$ plays the role of $p$. Let $j'$ be the index such
that $e = b^1_{i',j'}$. Since $i' < p$ it follows by induction that
$h^{1P}(F,B'',\sigma) \ge f^{1P}(N(F,p)\cap [i'-1],\hat\sigma)$.

\medskip\noindent
\textbf{Case 3:}
Assume that $e \in \bbb^1_{p}$ and that $\sigma(e) \ne
\sigma(\bbb^1_{p})$. In particular, we must have $\sigma(e) >
\sigma(\bbb^1_p)$ which means that the last requirement for $(i)$ is
violated such that condition $(ii)$ must be satisfied instead.
By an argument
analogous to the one given in case 1 $(i)$ we see that:
\begin{align*}
h^{1P}(F,B,\sigma) &~=~ h^{1P}(F\setminus\{e\},B,\sigma) + 1 +
h^{1P}(F,B'',\sigma)\;.
\end{align*}
Using the induction hypothesis, we show that
$h^{1P}(F,B'',\sigma)
\ge f^{1P}(N(F,p),\hat\sigma)$, i.e., we only count improving switches
performed during the second recursive call.

The remainder of the argument is the same as the last part of the
argument for case 1 $(i)$ where $i'$ is replaced by $p$.
Note, however, that in order to apply the induction hypothesis 
we need the observation that $|\bbb^1_p \cap B''| > |\bbb^1_p \cap
B|$. Note also that if $|\bbb^1_p \cap B| =
rs$ such that $\bbb^1_p \subseteq B$, then it was not possible to
pick $e \in \bbb^1_{p}$, and case 2 could not happen. The base-case
$|\bbb^1_p \cap B| = rs$ then follows from the proof of the other
four cases.

\medskip\noindent
\textbf{Case 4:}
Assume that 
$e \in \aaa^1_{p,j}$ for some $j\in[r]$ such that
$\aaa^1_{p} \not\sqsubseteq (F\setminus \{e\})\cup B$.
By an argument
analogous to the one given in case 1 $(i)$ we see that:
\begin{align*}
h^{1P}(F,B,\sigma) &~=~ h^{1P}(F\setminus\{e\},B,\sigma) + 1 +
h^{1P}(F,B'',\sigma)\;.
\end{align*}
Using the induction hypothesis, we show that
$h^{1P}(F,B'',\sigma)
\ge f^{1P}(N(F,p),\hat\sigma)$, i.e., we only count improving switches
performed during the second recursive call.

The most important observation for
this case is that $reset((F\setminus \{e\}) \cup B') = p$. To prove
this, observe that $\bbb^1_p \subseteq F \cup B'$. For $(i)$ this
follows from the fact that $V(b_{p,1})$ is fixed for $B$ w.r.t. $F$. For
$(ii)$ some edges in $\bbb^1_p$ are fixed and the remaining edges are
in $F$.
Lemma \ref{L-Sigma-optimal} then shows that $F$, $B''$, and $\sigma$ satisfy
$(i)$. Note that for both $(i)$ and $(ii)$ we must have $e \in
\aaa^1_{p,j'}$ where $\sigma(\aaa^1_{p,j'}) = \sigma(\aaa^1_{p})$.
Hence, $|\aaa^1_{p,j'} \cap B''| > |\aaa^1_{p,j'} \cap
B|$, and we get by induction that $h^{1P}(F,B'',\sigma)
\ge f^{1P}(N(F,p),\hat\sigma)$. Note that if $|\aaa^1_{p,j} \cap B| =
s$ such that $\aaa^1_{p,j} \subseteq B$, then it was not possible to
pick $e \in \aaa^1_{p,j} \setminus B$, and case 3 could not
happen. The base-case $|\aaa^1_{p,j} \cap B| = s$ then follows from
the proof of the other four cases.

\medskip\noindent
\textbf{Case 5:}
Assume that $e$ does not qualify for any of the first four cases.
We show that $h^{1P}(F\setminus\{e\},B,\sigma) \ge
f^{1P}(N(F,p),\hat\sigma)$, i.e., in this case we only count the improving
switches performed during the first recursive call.

Observe that $N(F\setminus\{e\},p) = N(F,p)$ since, due to cases 1 and
2, $e \not\in
\bbb^1_{i'}$ for some $i' \in N(F,p)$. We also claim that
$reset((F\setminus \{e\}) \cup B) = 0$. Assume for the sake of
contradiction that there exists an index $i\in[n]$ such that $\bbb^1_i
\subseteq (F\setminus \{e\}) \cup B$ and $\aaa^1_i
\not\sqsubseteq (F\setminus \{e\}) \cup B$. Then it must be the case
that $bit_i(F\cup B,B) \ne 1$, since otherwise the relevant edges
are part of $B$. On the other hand, we can also not have
$bit_i(F\cup B,B) = 0$ since $\sigma$ is well-behaved.
The only remaining possibility is $i = p$, but this case has been
ruled out due to case 4. The claim follows.
Observe also that due to cases 1, 2, and 4 we can not have $\ell =
\sigma(\bbb^1_p)$ or $\ell = \sigma(\aaa^1_p)$.
The remaining requirements for conditions $(i)$ and $(ii)$ are not
affected by removing $e$ from $F$, and it follows that if $F$,
$B$, and $\sigma$ satisfy $(i)$ or $(ii)$, respectively, then
$F\setminus \{e\}$, $B$, and $\sigma$ satisfy $(i)$ or $(ii)$,
correspondingly. It remains to show that $F \setminus \{e\}$ is
functional. This follows from $\sigma$ being well-behaved and $\ell
\le \sigma(\aaa^1_p)$.

Since $|F \setminus \{e\}| < |F|$ it follows by
induction that $h^{1P}(F\setminus\{e\},B,\sigma) \ge
f^{1P}(N(F,p),\hat\sigma)$. Note
that it is not possible to reach a situation where $F$ is not
functional. Hence, the induction always goes back to one of the other
base-cases. The base-case where $F$ is not functional therefore
follows from the proof of the other four cases.
\end{proof}

%% file: Conclude.tex
\section{Concluding remarks and open problems}\label{S-concluding-remarks}

We obtained an $2^{\tilde{\Omega}(\sqrt[3]{m})}$ lower bound for \RandomFacet. The question whether this can be improved to $2^{\tilde{\Omega}(\sqrt{m})}$, which would then be tight, is an interesting open problem. We also obtained $2^{\tilde{\Omega}(\sqrt{m})}$ lower bounds for \RandomPerm\ and \RandomBland. No subexponential upper bounds are currently known for these two algorithms. Obtaining such upper bounds, or further improving our lower bounds, are also interesting open problems.

Our lower bounds were all obtained using shortest paths problems. Another interesting open problem is whether similar lower bounds can be obtained using graphs in which the out-degree of each vertex is~$2$. Such \emph{binary} instances would supply explicit \emph{Acyclic Unique Sink Orientations} (AUSOs) on which \RandomFacet\ takes subexponential time. (For more on AUSOs, see G{\"a}rtner \cite{Gartner02}, Matou{\v{s}}ek \cite{Matousek94}, Matou{\v{s}}ek and Szab{\'o} \cite{MaSz06}, Schurr and Szab{\'o} \cite{ScSz04,ScSz05}, and Szab{\'o} and Welzl \cite{SzWe01}.)